\documentclass[a4paper,11pt,onecolumn,unpublished]{quantumarticle}
\pdfoutput=1

\usepackage[T1]{fontenc}
\usepackage[utf8]{inputenc}
\usepackage{amsmath,amssymb,amsthm,mathtools,mathrsfs}

\usepackage{enumitem}
\usepackage{microtype}
\usepackage{booktabs}
\usepackage{array}
\usepackage{graphicx}
\usepackage{tikz}
\usetikzlibrary{arrows.meta,positioning}
\usepackage[colorlinks=true,linkcolor=blue,citecolor=blue,urlcolor=blue]{hyperref}

\newtheorem{theorem}{Theorem}[section]
\newtheorem{proposition}[theorem]{Proposition}
\newtheorem{lemma}[theorem]{Lemma}
\newtheorem{corollary}[theorem]{Corollary}
\theoremstyle{definition}
\newtheorem{definition}[theorem]{Definition}

\newtheorem{remark}[theorem]{Remark}

\newcommand{\R}{\mathbb R}
\newcommand{\Z}{\mathbb Z}
\newcommand{\C}{\mathbb C}
\newcommand{\T}{\mathbb T}

\newcommand{\vol}{\operatorname{vol}}
\newcommand{\Ran}{\operatorname{Ran}}
\newcommand{\Span}{\operatorname{span}}

\newcommand{\cH}{\mathcal H}
\newcommand{\cC}{\mathcal C}
\newcommand{\cG}{\mathcal G}
\newcommand{\cS}{\mathcal S}
\newcommand{\cZ}{\mathcal Z}

\newcommand{\whatrho}{\widehat\rho}
\newcommand{\ip}[2]{\langle #1,#2\rangle}
\newcommand{\norm}[1]{\lVert #1\rVert}
\newcommand{\abs}[1]{\lvert #1\rvert}

\title{A Time--Frequency Framework for GKP Codes}
\author{Franz Luef}
\email{franz.luef@ntnu.no}
\affiliation{Department of Mathematical Sciences, Norwegian University of Science and Technology (NTNU), Trondheim, Norway}
\author{Eduard Ortega}
\email{eduard.ortega@ntnu.no}
\affiliation{Department of Mathematical Sciences, Norwegian University of Science and Technology (NTNU), Trondheim, Norway}

\date{}

\hypersetup{
  pdftitle={A Time--Frequency Framework for GKP Codes},
  pdfauthor={Franz Luef and Eduard Ortega}
}

\begin{document}
\maketitle

\begin{abstract}
\emergencystretch=1em
We develop a time--frequency framework for lattice GKP codes in which
ideal codewords are realized in the modulation space
$M^\infty$ and identified, through a vector-valued Zak transform,
with a finite logical fibre over the continuous syndrome torus.
Multi-window Gabor analysis then represents the logical vector by a
finite block of adjoint-lattice coefficients. We prove that the
normalized block map is an isometry, obtain an explicit recovering
projection, and derive stable logical reconstruction. We further
construct normalizable GKP approximants as lattice-envelope Gabor
multipliers and establish weak-$*$ convergence and asymptotically
isometric encoding. Finally, we recover displacement syndromes from
phase relations between translated coefficient blocks and quantify
their stability under additive perturbations.
\end{abstract}

\section{Introduction}
\label{sec:introduction}

Bosonic quantum error correction uses the infinite-dimensional Hilbert
space of one or more oscillators to protect finite-dimensional quantum
information. Among its most prominent constructions is the
Gottesman--Kitaev--Preskill (GKP) code, which encodes a logical qubit,
or more generally a finite-dimensional logical system, into highly
structured phase-space states \cite{GKP}. Lattice formulations of GKP
codes make explicit the stabilizer group, the logical displacement
group, the symplectic geometry of the code, and the geometry of
displacement decoding; see, for instance,
\cite{ConradEisertArzani}. Ideal GKP codewords are not normalizable
oscillator states, but distributional grid states. More precisely, the natural distribution space for ideal GKP
codewords is the modulation space \(M^\infty(\mathbb R^n)\), the
continuous dual of the Feichtinger algebra
\(M^1(\mathbb R^n)\). This choice is essential. Although
\(M^\infty(\mathbb R^n)\) is contained in the space of tempered
distributions, working in the whole of \(\mathcal S'(\mathbb R^n)\)
would allow additional lattice-supported distributions, including
derivatives of Dirac masses and coefficients with polynomial growth.
The \(M^\infty\) condition restricts the Zak transform of a stabilized
state to a Dirac comb with uniformly bounded coefficients. The Zak
transition relations then determine the entire comb from a single
vector in the finite-dimensional logical fibre. This gives a genuine
bijection between ideal GKP codewords and their logical vectors,
while retaining the natural duality with the localized probes in
\(M^1(\mathbb R^n)\).

Zak and modular-variable coordinates are particularly well adapted to
this lattice structure. They separate the continuous displacement
syndrome from a finite-dimensional logical fibre; see
\cite{PantaleoniBaragiolaMenicucci}. A complementary geometric
description in terms of polarized abelian varieties, theta functions,
and Heisenberg modules has recently been developed in
\cite{MayrandRoyer}. For a related approach using quantum tori and Heisenberg
modules, see \cite{JosephSingh2025}. The purpose of the present work is not to replace
these descriptions, but to connect them with a concrete
time--frequency coefficient model. Our guiding question is:
\begin{quote}
\emph{How is the finite-dimensional logical information of an ideal
GKP codeword represented by localized phase-space coefficients, and
how stably can it be reconstructed from finitely many such
coefficients?}
\end{quote}

This question leads directly to Gabor analysis. Weyl displacement
operators are precisely the time--frequency shifts of harmonic
analysis. Moreover, if
\(\Lambda\subset\mathbb R^{2n}\) is a symplectically integral
stabilizer lattice, then its symplectic adjoint
\[
\Lambda^\circ
=
\bigl\{
z\in\mathbb R^{2n}:
\sigma(\lambda,z)\in\mathbb Z
\text{ for every }\lambda\in\Lambda
\bigr\}
\]
is simultaneously the lattice of displacements commuting with the
GKP stabilizers and the adjoint lattice appearing in Gabor duality.
Thus
\[
\Lambda\subset\Lambda^\circ
\]
has two parallel interpretations: \(\Lambda\) is the stabilizer and
Gabor lattice, while \(\Lambda^\circ\) is the logical centralizer and
the adjoint time--frequency lattice.

Fix a stabilizer phase \(\chi:\Lambda\to\mathbb T\), and write
\[
S_\lambda=\chi(\lambda)\rho(\lambda),
\qquad \lambda\in\Lambda,
\]
for the corresponding commuting stabilizer operators. The associated ideal
GKP code space is
\[
\mathcal C_{\Lambda,\chi}
:=
\left\{
\psi\in M^\infty(\mathbb R^n):
S_\lambda\psi=\psi
\text{ for every }\lambda\in\Lambda
\right\}.
\]
For the rectangular normal form \(\Lambda_D\), equipped with its canonical
phase \(\chi_D\), we write simply
\(\mathcal C_{\Lambda_D}\).
As proved in Section~4, this code space is naturally isomorphic to the
finite-dimensional logical Hilbert space \(\mathcal H_{\mathbf d}\).
For \(c\in\mathcal H_{\mathbf d}\), we denote by
\(\psi_c\in\mathcal C_{\Lambda_D}\) the corresponding ideal codeword.

The quotients 
\[
K_\Lambda:=\Lambda^\circ/\Lambda
\qquad \text{and}\qquad\mathbb R^{2n}/\Lambda^\circ
 \]
describe the logical and syndrome
variables. Each class
\(q\in K_\Lambda\) determines a logical displacement and, after fixing
the stabilizer phases, a finite Weyl operator
\[
\mathsf W_q:\mathcal H_{\mathbf d}
\longrightarrow\mathcal H_{\mathbf d}.
\]
A probe selects a vector \(v_\alpha\in\mathcal H_{\mathbf d}\), and
the corresponding finite Weyl orbit
\[
\bigl\{
\mathsf W_qv_\alpha:q\in K_\Lambda
\bigr\}
\]
provides a finite family of analyzing vectors for the logical state.
The central point of the paper is that the resulting logical
coefficients
\[
\left\langle c,\mathsf W_qv_\alpha\right\rangle
\]
appear exactly as a finite block of adjoint-lattice Gabor coefficients
sampled from the encoded codeword.

To make this correspondence explicit, we first work in a rectangular
symplectic normal form
\[
\Lambda_D
=
D^{1/2}\mathbb Z^n\times D^{1/2}\mathbb Z^n,
\]
where
\[
D=\operatorname{diag}(d_1,\ldots,d_n),
\qquad
d_1\mid\cdots\mid d_n.
\]
We write
\[
\mathcal H_{\mathbf d}
:=
\bigotimes_{j=1}^n\mathbb C^{d_j},
\qquad
d:=\dim\mathcal H_{\mathbf d}
=\prod_{j=1}^n d_j
=\operatorname{vol}(\Lambda_D).
\]
The vector-valued Zak transform identifies the ideal codeword
associated with \(c\in\mathcal H_{\mathbf d}\) with the
distributional section
\[
\mathcal Z_D\psi_c=c\,\delta_0.
\]
Thus, the logical vector is concentrated at the origin of the syndrome
torus, while displacement errors translate this support.

The first part of the paper establishes this dictionary. It describes
the logical Pauli action through the finite Weyl representation, shows
how suitable lattice-preserving metaplectic transformations induce
logical Clifford operations, and explains how displacement errors act
on the syndrome variable. We also construct ideal codewords from
multi-window Gabor frames and extend the rectangular description to
arbitrary symplectically integral lattices by phase-corrected
metaplectic transport. These ingredients are standard or close to
standard in the GKP and Gabor literature; our role is to place them
within a common set of conventions suited to the coefficient model
developed below.

\subsection*{Finite coefficient blocks and logical reconstruction}

The new point begins with the adjoint-lattice coefficient
representation. Given \(c\in\mathcal H_{\mathbf d}\), let \(\psi_c\)
be the corresponding ideal GKP codeword. A multi-window Gabor frame
associates with \(\psi_c\) a family of time--frequency coefficients
indexed by \(\Lambda_D^\circ\). Stabilizer covariance implies that
coefficients whose indices differ by an element of \(\Lambda_D\) are
related by explicitly known phases. Consequently, all independent
logical information is contained in one finite block indexed by
\[
K_D=\Lambda_D^\circ/\Lambda_D.
\]

Our main result, Theorem~\ref{thm:finite-block-recovery}, shows that
this finite block determines the logical vector completely. Choose
representatives
\[
\nu_q\in\Lambda_D^\circ,
\qquad q\in K_D,
\]
and let \(h_\alpha\) be the probe determined by the multi-window
Gabor system. The entries of the fundamental block are samples of the
Gabor coefficient representation of the ideal codeword:
\[\frac{1}{d}
\left\langle
\psi_c,\rho(\nu_q)h_\alpha
\right\rangle,
\qquad q\in K_D.
\]
Thus the block is obtained on the signal side by sampling the
time--frequency coefficient transform of \(\psi_c\) at one
representative of every class in
\(\Lambda_D^\circ/\Lambda_D\).

Let
\[
v_\alpha:=\mathcal Z_Dh_\alpha(0)
\in\mathcal H_{\mathbf d}
\]
be the logical-fibre vector selected by the probe, and assume that
\(v_\alpha\neq0\). Define the normalized fundamental block by
\[
B_0(c)
:=
\frac{\sqrt d}{\|v_\alpha\|}
\left(
\left\langle
\psi_c,\rho(\nu_q)h_\alpha
\right\rangle
\right)_{q\in K_D}.
\]
Theorem~\ref{thm:finite-block-recovery} identifies these sampled
signal-side coefficients exactly with finite Weyl analysis
coefficients of \(c\):
\begin{equation}
\label{eq:intro-main-formula}
\underbrace{
B_0(c)_q
}_{\substack{\text{normalized Gabor sample}\\
             \text{of the encoded codeword}}}
=
\underbrace{
\frac{1}{\sqrt d\,\|v_\alpha\|}
\left\langle c,\mathsf W_qv_\alpha\right\rangle
}_{\substack{\text{finite Weyl coefficient}\\
             \text{of the logical vector}}},
\qquad q\in K_D.
\end{equation}
In other words, sampling the Gabor coefficient representation of the
encoded distribution and analyzing the logical vector along the
finite Weyl orbit of \(v_\alpha\) produce the same finite block.

It follows that
\[
B_0:\mathcal H_{\mathbf d}\longrightarrow\ell^2(K_D)
\]
is an isometry. Therefore, the sampled fundamental block determines
the logical vector uniquely and stably through
\begin{equation}
\label{eq:intro-reconstruction}
c
=
\frac{1}{\sqrt d\,\|v_\alpha\|}
\sum_{q\in K_D}
B_0(c)_q\,\mathsf W_qv_\alpha.
\end{equation}
The stabilizer covariance relations then determine every remaining
adjoint-lattice coefficient from this single block, up to explicitly
known phases. Schematically,
\[
c
\longrightarrow
\psi_c
\longrightarrow
\text{adjoint-lattice Gabor coefficients}
\longrightarrow
B_0(c)
\longrightarrow
c.
\]

The isometric embedding also identifies the logical Hilbert space
with the closed subspace
\[
\mathcal B_0:=\operatorname{Ran}(B_0)
\subset\ell^2(K_D).
\]
The operator
\[
\Pi_0:=B_0B_0^*
\]
is the orthogonal projection onto the space of admissible exact
blocks. Hence, for arbitrary coefficient data
\(x\in\ell^2(K_D)\), the vector $\Pi_0x$ 
is the closest exact logical block, and
\[
B_0^*x=B_0^*\Pi_0x
\]
is the corresponding reconstructed logical vector. In particular, $\|x-\Pi_0x\|_{\ell^2(K_D)}$
quantifies the inconsistency of the data with the exact finite-block
model.

This finite-block identification and its recovering projection form
the central contribution of the paper. In addition, the probe can be
chosen to realize any prescribed nonzero vector in the logical fibre.
For suitable Gaussian probes, the data can be reduced further:
certain subfamilies containing only \(d=\operatorname{vol}(\Lambda_D)\)
coefficients already determine the logical vector. This is the
minimum possible number of scalar complex coefficients required to
reconstruct an arbitrary vector in the \(d\)-dimensional logical
space.

\subsection*{Normalizable regularizations}

The coefficient formulation also separates the finite logical
information from the infinite stabilizer orbit. The ideal codeword
admits a weak-* expansion
\[
\psi_c
=
\sum_{\lambda\in\Lambda}
S_\lambda\Gamma_c,
\]
where \(\Gamma_c\) contains the logical coefficients and the
stabilizer orbit produces the ideal lattice-periodic state. Replacing
the constant lattice weight by an envelope
\(w=(w_\lambda)_{\lambda\in\Lambda}\) gives
\[
\Psi_c^w
=
\sum_{\lambda\in\Lambda}
w_\lambda S_\lambda\Gamma_c.
\]
This regularization is naturally a Gabor multiplier: the vector
\(\Gamma_c\) carries the logical information, whereas the multiplier
symbol \(w\) controls the stabilizer envelope.

For square-summable envelopes, the resulting states are
normalizable. For families of envelopes approaching the constant
symbol, we prove weak-* convergence to the ideal codeword. Moreover,
after an explicit normalization, asymptotically
translation-invariant envelopes yield asymptotically isometric
logical embeddings. This provides a general criterion for
lattice-envelope Gabor multipliers and complements the canonical
number-operator regularization studied in
\cite[Theorem~5.2]{MayrandRoyer}. We also prove convergence of the
regularized finite blocks to their ideal counterparts, and hence
stable logical reconstruction from sufficiently accurate
normalizable approximations.

\subsection*{Syndrome reconstruction}

Finally, displacement syndromes appear as characters of the
stabilizer lattice. In coefficient space, this character becomes a
phase relation between translated finite blocks. It can therefore be
reconstructed by correlating such blocks, and the resulting quotient
satisfies an explicit deterministic stability estimate under
additive coefficient perturbations. Combining several stabilizer
directions gives an overdetermined, and optionally weighted,
least-squares estimator of the displacement class in the syndrome
torus. The block-covariance relations and their correlation estimator
also persist asymptotically for the normalizable regularizations.

Thus the same coefficient representation has two complementary
roles: one fundamental block encodes the finite logical vector, while
the relations between translated blocks encode the continuous
displacement syndrome. In this sense,
\emph{logical information, syndrome information, and
time--frequency coefficients are different aspects of the same
lattice structure}.

\paragraph{Scope of the coefficient model.}
The quantities considered below are complex time--frequency
coefficient functionals
\[
a_{\mu,\alpha}(F)
=
\frac1d
\left\langle F,\rho(\mu)h_\alpha\right\rangle.
\]
Our covariance, reconstruction, and perturbation results are
mathematical statements conditional on access to these complex
amplitudes. They do not, by themselves, specify a physical quantum
measurement or recovery channel. A direct rank-one measurement
ordinarily produces a probability depending on
\[
\left|
\left\langle F,\rho(\mu)h_\alpha\right\rangle
\right|^2,
\]
whereas the block-correlation formulas also require relative complex
phases. Accessing these phases through coherent-reference,
ancilla-assisted, or tomographic procedures, and developing the
corresponding statistical noise model, lie outside the scope of this
paper. Accordingly, ``reconstruction'' means reconstruction from
coefficient data, and ``coefficient noise'' means deterministic
additive perturbation of those data.

\paragraph{Organization of the paper.}
Section~2 introduces symplectically integral lattices and stabilizer
phases. Section~3 recalls the required
Gabor-frame, modulation-space, and Heisenberg-module machinery.
Section~4 introduces  GKP codes associated to symplectically integral lattices and their metaplectic transformations. It develops the vector-valued Zak description of the code,
including logical Pauli and Clifford operations, displacement
syndromes, the construction of ideal codewords from Gabor frames, and
the passage to general symplectically integral lattices. Section~5
introduces the adjoint-lattice coefficient model and proves the
finite-block reconstruction theorem, together with the recovering
projection, probe design, minimal sampling, the extension to general
lattices, and an explicit single-mode example. Section~6 develops
normalizable lattice-envelope regularizations as Gabor multipliers,
proves asymptotically isometric normalizable encoding, and establishes
logical reconstruction from regularized coefficients. Section~7
identifies the syndrome character through translated blocks, develops
block-correlation and redundant syndrome estimators, proves stability
under coefficient perturbations, and relates these results to the
regularized states of Section~6.

\section{GKP lattices, logical fibres, and syndrome space}
\label{sec:gkp-background}

This section fixes conventions and recalls the pieces of lattice GKP theory
needed later. The results are standard, or immediate adaptations to our Weyl
normalization; see
\cite{GKP,ConradEisertArzani,PantaleoniBaragiolaMenicucci,MayrandRoyer}.

Throughout the paper, the Hilbert-space inner product is linear in
the first variable:
\[
\langle af,g\rangle=a\langle f,g\rangle,
\qquad
\langle f,ag\rangle=\overline a\,\langle f,g\rangle.
\]

\subsection{Weyl operators and symplectically integral lattices}

Write \(\Xi=\R^n\times\R^n\) and \(z=(x,\omega)\). Let
\[
(T_xf)(t)=f(t-x),\qquad
(M_\omega f)(t)=e^{2\pi i\omega\cdot t}f(t), \qquad f\in L^2(\mathbb R^n)
\]
and define the unitary \emph{Weyl operator} on $L^2(\mathbb R^n)$
\[
\rho(x,\omega)=e^{-\pi ix\cdot\omega}M_\omega T_x.
\]
For
\[
J=\begin{pmatrix}0&I_n\\-I_n&0\end{pmatrix},
\qquad\text{and}
\qquad\sigma(z,w)=z^TJw,
\]
one has the following relations 
\begin{align}
\rho(z)\rho(w)&=e^{-\pi i\sigma(z,w)}\rho(z+w), \label{eq:weyl-commutation_0}
\\
\rho(z)\rho(w)&=e^{-2\pi i\sigma(z,w)}\rho(w)\rho(z),
\label{eq:weyl-commutation}\\
\rho(z)^*&=\rho(-z).\nonumber
\end{align}

A full lattice of $\mathbb R^{2n}$ is a discrete subgroup  $\Lambda=A\Z^{2n}$ with
$A\in GL(2n,\R)$, then define its covolume $\vol(\Lambda)=|\det A|$ and symplectic Gram matrix $\Theta_A=A^TJA$.
If the lattice basis is changed by a matrix
\(P\in\operatorname{GL}(2n,\mathbb Z)\), then \(A\) is replaced by
\(A'=AP\). Hence $\Lambda=A'\mathbb Z^{2n}
=A\mathbb Z^{2n}$.
The corresponding symplectic Gram matrix transforms according to
\[
\Theta_{A'}
=(A')^{T}JA'
=P^{T}A^{T}JAP
=P^{T}\Theta_A P.
\]
Thus, a change of lattice basis leaves \(\Lambda\) unchanged and replaces
\(\Theta_A\) by the integrally congruent matrix \(P^{T}\Theta_A P\).
\begin{definition}
A full lattice $\Lambda$ is \emph{symplectically integral} if $\sigma(\Lambda,\Lambda)\subset\Z$.
Equivalently, $\Theta_A$ has integer entries.
\end{definition}

By \eqref{eq:weyl-commutation}, integrality is precisely the condition that
the Weyl operators associated with lattice points commute.   

If $\Lambda$ is symplectically integral, then given $A\in GL(2n,\mathbb R)$ with $\Lambda=A\mathbb Z^{2n}$ the Frobenius normal form yields
\[
 M^T\Theta_AM=
 \begin{pmatrix}0&D\\-D&0\end{pmatrix},
 \qquad
 D=\operatorname{diag}(d_1,\ldots,d_n),
 \qquad d_1\mid\cdots\mid d_n,
\]
for some $M\in GL(2n,\Z)$. Put
\[
 A_D=\begin{pmatrix}\sqrt D&0\\0&\sqrt D\end{pmatrix},
 \qquad
 \Lambda_D=A_D\Z^{2n},\qquad \text{(square lattice)}
\]
then $\Theta_{A_D}=
 \begin{pmatrix}0&D\\-D&0\end{pmatrix}$.
The following lemma is well known.
\begin{lemma}\label{lem:same-gram}
For $A,B\in GL(2n,\R)$, the equality $\Theta_A=\Theta_B$ holds if and only if
$A=SB$ for some $S\in\operatorname{Sp}(2n,\R)$.
\end{lemma}

\begin{proposition}[Symplectic normal form]\label{prop:normal-form}
Every full symplectically integral lattice is of the form $\Lambda=S\Lambda_D$ for a symplectic matrix $S$ and a uniquely determined divisibility tuple
$(d_1,\ldots,d_n)$.
\end{proposition}

\begin{proof}
Replace $A$ by $AM$ in the skew Smith normal form and apply
Lemma~\ref{lem:same-gram} to $AM$ and $A_D$.
\end{proof}

In particular,
\[
 d:=\vol(\Lambda_D)=\det D=\prod_{j=1}^nd_j\in\mathbb N.
\]

The adjoint lattice is
\begin{equation*}
\Lambda^\circ
=\{z\in\Xi:\sigma(\lambda,z)\in\Z
\text{ for every }\lambda\in\Lambda\}.
\end{equation*}
For \(\Lambda=A\Z^{2n}\),
\[
\Lambda^\circ=-JA^{-T}\Z^{2n}.
\]
Thus, if  \(\Lambda\) is symplectically integral, then \(\Lambda\subseteq\Lambda^\circ\). For the
square lattice,
\begin{equation*}
\Lambda_D^\circ=
\begin{pmatrix}D^{-1/2}&0\\0&D^{-1/2}\end{pmatrix}\Z^{2n},
\qquad
\abs{\Lambda_D^\circ/\Lambda_D}=d^2.
\end{equation*}

\subsection{Stabilizer phases}

\begin{definition}\label{def:phase}
A \emph{stabilizer phase} for an integral lattice $\Lambda$ is a map
$\chi:\Lambda\to\T$ satisfying
\begin{equation}\label{eq:chi-cocycle}
 \chi(\lambda+\mu)
 =\chi(\lambda)\chi(\mu)e^{-\pi i\sigma(\lambda,\mu)}.
\end{equation}
The corresponding stabilizer operators are
\[
 S^\chi_\lambda=\chi(\lambda)\rho(\lambda).
\]
\end{definition}
Equation \eqref{eq:chi-cocycle} implies that $S^\chi_\lambda S^\chi_\mu=S^\chi_{\lambda+\mu}$ for $\lambda,\mu\in \Lambda$. Thus, $\lambda\mapsto S^\chi_\lambda$ is a unitary representation of
the abelian group $\Lambda$. If $\chi$ is understood, then we will denote the operators $S^\chi$ just by $S$.

\begin{proposition}\label{prop:phase-exists}
Every symplectically integral lattice admits a stabilizer phase.
\end{proposition}

\begin{proof}
Choose an ordered basis $b_1,\ldots,b_{2n}$ of $\Lambda$ and write
$\lambda=\sum_jk_jb_j$.  Put
\[
 q(k)=\sum_{i<j}\sigma(b_i,b_j)k_ik_j\pmod 2,
 \qquad
 \chi(\lambda)=(-1)^{q(k)}.
\]
Since $\sigma(b_i,b_j)\in \mathbb Z$,
\[
 q(k+\ell)-q(k)-q(\ell)
 =\sigma\left(\sum_ik_ib_i,\sum_j\ell_jb_j\right)\pmod2.
\]
This is exactly \eqref{eq:chi-cocycle}.
\end{proof}

For the rectangular lattice $\Lambda_D
=
\sqrt D\,\mathbb Z^n
\times
\sqrt D\,\mathbb Z^n$,
we fix the canonical stabilizer phase
\[
\chi_D(\sqrt D\,k,\sqrt D\,\ell)
:=
e^{-\pi i k^T D\ell}
=
(-1)^{k^T D\ell},
\qquad k,\ell\in\mathbb Z^n.
\]
In particular,
\[
\chi_D(-\lambda)=\chi_D(\lambda)\in\{\pm1\}.
\]

\section{Gabor frames and Heisenberg modules}
\label{sec:gabor}

We now collect the time--frequency tools that will be used to pass from the Zak description of the code to stable adjoint-lattice coefficients. We recall modulation spaces and multi-window Gabor frames, together with the duality principle and Wexler--Raz biorthogonality; see \cite{GrochenigBook,JakobsenLuef}. We then describe the associated Heisenberg-module structure and the module idempotent arising from a Gabor frame, following \cite{Luef2009,JakobsenLuef}. Finally, we record the metaplectic covariance of Gabor systems, which will allow the later coefficient constructions to be transported between symplectically equivalent lattices.

\subsection{Modulation spaces}

Modulation spaces were introduced by Feichtinger in the early 1980s;
see \cite{FeichtingerModulationSpaces}. For a nonzero $g$ in the Schwartz space $\cS(\R^n)$ and a tempered distribution $f\in \cS'(\R^n)$ we define the short-time Fourier transform (STFT)
\[
 V_gf(z)=\langle f,\rho(z)g\rangle.
\]
For $1\le p\le\infty$, the modulation space $M^p(\R^n)$ consists of the
tempered distributions $f$ for which $V_gf\in L^p(\R^{2n})$.  Different
Schwartz windows give equivalent norms. For $1\leq p<\infty$ and
$p^{-1}+q^{-1}=1$, the standard duality is
\[
 (M^p(\R^n))'=M^q(\R^n),
\]
with the weak-* convention
$M^\infty(\R^n)=(M^1(\R^n))'$. The Weyl operators act isometrically on
the modulation spaces.

If $f,g\in M^1(\R^n)$, then lattice samples of their STFT are absolutely
summable.  In particular,
\begin{equation}\label{weak_conv}
 \sum_{\lambda\in\Lambda}
 |\langle f,\rho(\lambda)g\rangle|
 \le C_{\Lambda,g}\|f\|_{M^1},
\end{equation}
for some $C_{\Lambda,g}>0$.
This estimate makes
 lattice orbit sums meaningful as weak-* elements of
$M^\infty(\mathbb R^n)$.

\subsection{Multi-window frames and super Riesz families}

Let $\mathbf g=(g_1,\ldots,g_m)\in L^2(\R^n)\otimes\C^m$.  The system
\[
 \cG(\mathbf g,\Lambda)
 =\{\rho(\lambda)g_j:\lambda\in\Lambda, 1\le j\le m\}
\]
is called a \emph{multi-window Gabor frame} if there exists $A,B>0$ such that 
\[
 A\|f\|_2^2
 \le\sum_{j=1}^m\sum_{\lambda\in\Lambda}
 |\langle f,\rho(\lambda)g_j\rangle|^2
 \le B\|f\|_2^2, \qquad f\in L^2(\mathbb R^n).
\]
Then the \emph{frame operator} 
\[
 S^\Lambda_{\mathbf g}f
 :=\sum_{j=1}^m\sum_{\lambda\in\Lambda}
 \langle f,\rho(\lambda)g_j\rangle\rho(\lambda)g_j, \qquad f\in L^2(\mathbb R^n),
\]
is invertible.
The canonical dual family $\boldsymbol  \gamma=(\gamma_1,\ldots,\gamma_m)\in L^2(\R^n)\otimes\C^m$ is given by 
\[
 \gamma_j=(S^\Lambda_{\mathbf g})^{-1}g_j \qquad j=1,\ldots,m,
\]
and so we have 
\[
 f=\sum_{j=1}^m\sum_{\lambda\in\Lambda}
 \langle f,\rho(\lambda)g_j\rangle\rho(\lambda)\gamma_j, \qquad f\in L^2(\mathbb R^n).
\]
If $\mathbf g\in M^1(\mathbb R^n)\otimes \mathbb C^m$, spectral invariance implies that the
$\boldsymbol\gamma$ also belong to $M^1(\mathbb R^n)\otimes \mathbb C^m$ (see \cite{GrochenigLeinert}), and then the above reconstruction formula extends to any modulation space,
\[
 f=\sum_{j=1}^m\sum_{\lambda\in\Lambda}
 \langle f,\rho(\lambda)g_j\rangle\rho(\lambda)\gamma_j, \qquad f\in M^p(\mathbb R^n),
\]
for $p\in [1,\infty]$; see \cite{GrochenigNoIneq}.

For  $\mathbf f=(f_1,\ldots,f_m)\in M^1(\R^n)\otimes\C^m$ define
\[
 \whatrho(z)\mathbf f
 =(\rho(z)f_1,\ldots,\rho(z)f_m) \qquad z\in \mathbb R^{2n}.
\]
The multi-window duality principle says that
$\cG(\mathbf g,\Lambda)$ is a frame if and only if
\[
 \mathcal R(\boldsymbol\gamma,\Lambda^\circ)
 =\{\whatrho(\mu)\boldsymbol\gamma:\mu\in\Lambda^\circ\}
\]
is a Riesz sequence
\cite{GrochenigNoIneq,JakobsenLuef}, that is, there exist $A,B>0$
\begin{equation}
\label{eq:model-space-stability}
A\|a\|_{\ell^p}
\leq
\left\|
\sum_{\mu\in\Lambda^\circ}
a_\mu\widehat\rho(\mu)\boldsymbol\gamma
\right\|_{M^p(\mathbb R^n)\otimes\mathbb C^m}
\leq
B\|a\|_{\ell^p}, \qquad a\in \ell^p(\Lambda^\circ).
\end{equation}

For canonical dual generators, the Wexler--Raz
relations read
\begin{equation*} 
 \sum_{j=1}^m
 \langle\gamma_j,\rho(\mu)g_j\rangle
 =\vol(\Lambda)\delta_{\mu,0},
 \qquad \mu\in\Lambda^\circ,
\end{equation*}
see \cite{JakobsenLuef}. Therefore the families
\[
 \{\whatrho(\mu)\boldsymbol\gamma\}_{\mu\in\Lambda^\circ}
 \quad\text{and}\quad
 \left\{\vol(\Lambda)^{-1}\whatrho(\mu)\mathbf g
 \right\}_{\mu\in\Lambda^\circ}
\]
are biorthogonal.

The density condition for multi-window Gabor frames gives $
m\geq \operatorname{vol}(\Lambda)$, see \cite[Lemma~4.9]{JakobsenLuef}.

\subsection{Heisenberg modules and the module idempotent}

We are now going to show the construction of Heisenberg modules following \cite{Luef2009,Rieffel}.  Let $A_\Lambda=C^*(\Lambda,c)$ be the twisted group algebra associated
with the cocycle
$c(\lambda,\mu)=e^{-\pi i\sigma(\lambda,\mu)}$,
that is, the operator-norm completion of the image of the  Weyl representation
\[
\rho_\Lambda:\ell^1(\Lambda,c)\to B(L^2(\mathbb R^n)),\qquad  a\longmapsto\sum_{\lambda\in\Lambda}a(\lambda)\rho(\lambda),
 \qquad a\in\ell^1(\Lambda,c).
\]
We also define  
\[
 c^\circ(\mu,\nu)
 :=\overline{c(\mu,\nu)}
 =e^{\pi i\sigma(\mu,\nu)},
 \qquad \mu,\nu\in\Lambda^\circ.
\]
and the corresponding twisted algebra $
 A_{\Lambda^\circ}:=C^*(\Lambda^\circ,c^\circ)$, 
represented through the adjoint Weyl operators
$\rho(\mu)^*=\rho(-\mu)$.

The space $M^1(\R^n)$ is a pre-equivalence bimodule between the twisted
algebras associated with $\Lambda$ and $\Lambda^\circ$. The left and right  inner product are
\[
 {}_\Lambda\!\langle f,h\rangle
 =\sum_{\lambda\in\Lambda}
 \langle f,\rho(\lambda)h\rangle \rho(\lambda),
\qquad \langle f,h\rangle_{\Lambda^\circ}
 =\frac1{\vol(\Lambda)}
 \sum_{\mu\in\Lambda^\circ}
 \langle h,\rho(\mu)^*f\rangle \rho(\mu)^*,\]
for $f,h\in M^1(\mathbb R^n)$.
The fundamental identity of Gabor analysis (FIGA) is the associativity relation
\[
 {}_\Lambda\!\langle f,g\rangle\cdot h
 =f\cdot\langle g,h\rangle_{\Lambda^\circ},
\]
for $f,g,h\in L^2(\mathbb R^n)$.

For a multi-window frame and its canonical dual, the matrix
\[
 \mathsf P_{\boldsymbol \gamma, \mathbf g}=\bigl({}_\Lambda\!\langle\gamma_r,g_s\rangle\bigr)_{r,s=1}^m
 \in M_m(A_\Lambda)
\]
is an idempotent \cite{Luef2009}.  It is not generally self-adjoint.  If an orthogonal
projection is desired, one may replace the original frame by the Parseval
family $(S^\Lambda_{\mathbf g})^{-1/2}\mathbf g$.

Recent work of Caragea, Kolountzakis, and Pfander proves that, whenever $\operatorname{vol}(\Lambda)<m$, there exist windows
$\mathbf g=(g_1,\ldots,g_m)\in\mathcal S(\mathbb R^n)\otimes\mathbb C^m$ 
such that \(\mathcal G(\mathbf g,\Lambda)\) is a multi-window Gabor frame;
see \cite[Theorems~1.3 and~1.9]{CarKolPfa26}.
Together with spectral invariance
\cite{GrochenigLeinert} and the module-frame correspondence
\cite{Luef2009,JakobsenLuef}, these results yield smooth
projections in the corresponding matrix algebras over the smooth
noncommutative torus.

\section{GKP codes and the Zak transform}

In this section we define the ideal GKP code space associated with a
symplectically integral lattice and a choice of stabilizer phase. The Zak
transform provides natural coordinates for this space: it separates the
continuous displacement syndrome from the finite-dimensional logical
degree of freedom; see~\cite{PantaleoniBaragiolaMenicucci} and, for the underlying time--frequency
theory,~\cite[Chapter~8]{GrochenigBook}.

We begin with the rectangular normal form \(\Lambda_D\). After recalling
the scalar Zak transform, we introduce its vector-valued form and prove
that it identifies the ideal code space with the logical fibre
\(\mathcal H_{\mathbf d}\) supported at the trivial point of the syndrome
torus. In these coordinates, adjoint-lattice displacements induce the
logical Pauli operators, suitable lattice-preserving metaplectic
transformations induce Clifford operations, and displacement errors
translate the syndrome variable. We then construct the ideal codewords
from multi-window Gabor frames and extend the entire description to
arbitrary symplectically integral lattices by phase-corrected metaplectic
transport.

For the remainder of the paper, fix
\[
D=\operatorname{diag}(d_1,\ldots,d_n),
\qquad
d_1\mid\cdots\mid d_n,
\qquad
d=d_1\cdots d_n,
\]
and set
\[
\Gamma_D=D^{1/2}\mathbb Z^n,
\qquad
\Gamma_D^\circ=D^{-1/2}\mathbb Z^n.
\]
Let \(Q_D\subseteq\mathbb R^n\) and
\(Q_D^\circ\subseteq\mathbb R^n\) be fundamental domains for
\(\Gamma_D\) and \(\Gamma_D^\circ\), respectively. We use Lebesgue
measure on the spatial fundamental domains and normalized Haar measure
\(d\omega\) on \(\mathbb R^n/\Gamma_D^\circ\), so that $\int_{Q_D^\circ}d\omega=1$.

% ============================================================
\subsection{The scalar Zak transform}
\label{subsec:scalar-zak}
% ============================================================

For \(f\in\mathcal S(\mathbb R^n)\), the scalar
\emph{Zak transform} associated with \(\Gamma_D\) is defined by
\begin{equation*}
\label{eq:scalar-zak}
Z_Df(x,\omega)
:=
\sqrt d\sum_{\gamma\in\Gamma_D}
f(x-\gamma)e^{2\pi i\gamma\cdot\omega}.
\end{equation*}
The Zak transform  satisfies the \emph{quasi periodicity} relations
\begin{align}
Z_Df(x+\gamma,\omega)
&=
e^{2\pi i\gamma\cdot\omega}Z_Df(x,\omega),
&&
\gamma\in\Gamma_D,
\label{eq:zak-qpx}
\\
Z_Df(x,\omega+\gamma^\circ)
&=
Z_Df(x,\omega),
&&
\gamma^\circ\in\Gamma_D^\circ.
\label{eq:zak-qpw}
\end{align}
With our normalization of Haar measure, \(Z_D\) extends uniquely to
a unitary operator
\[
Z_D:
L^2(\mathbb R^n)
\longrightarrow
L^2(Q_D\times Q_D^\circ);
\]
see \cite[Chapter~8]{GrochenigBook}. By duality we can extend  the Zak transform to $\mathcal S'(\mathbb R^n)$; see \cite[Chapter~8]{GrochenigBook}. We now describe the Zak-transform images of
\(M^1(\mathbb R^n)\) and \(M^\infty(\mathbb R^n)\); see
\cite{ToftZak}. Fix a nonzero window
\(\Phi\in\mathcal S(\mathbb R^{2n})\). For \(p\in\{1,\infty\}\), let
\[
W^{\infty,p}_{\mathrm{qp},\Gamma_D,\Gamma^\circ_D }(\mathbb R^{2n})
\]
be the space of all \(F\in\mathcal S'(\mathbb R^{2n})\) satisfying
\begin{align}
F(x+\gamma,\omega)
&=
e^{2\pi i\gamma\cdot\omega}F(x,\omega),
&&\gamma\in\Gamma_D,
\label{eq:wie-qpx}
\\
F(x,\omega+\gamma^\circ)
&=
F(x,\omega),
&&\gamma^\circ\in\Gamma_D^\circ,
\label{eq:wie-qpw}
\end{align}
in the distributional sense, and for which
\[
\|F\|_{W^{\infty,p}_{\mathrm{qp},\Gamma_D,\Gamma^\circ_D }}
:=
\|V_\Phi F\|_{
L^p((Q_D\times Q_D^\circ)\times\mathbb R^{2n})}
<\infty.
\]
For \(p=\infty\), the norm is interpreted as an essential
supremum. Different admissible windows give equivalent norms.

 For \(p=1\), the elements of $W^{\infty,1}_{\mathrm{qp},\Gamma_D,\Gamma^\circ_D }
  (\mathbb R^{2n})$
have continuous representatives, and the transition relations
\eqref{eq:wie-qpx} and \eqref{eq:wie-qpw}  hold pointwise. For \(p=\infty\), the elements are
understood as distributions in $M^\infty(\mathbb R^{2n})$.

Then by
\cite[Theorem~2.13 and Corollary~2.14]{ToftZak}, \(Z_D\) extends to
topological isomorphisms
\begin{equation}
\label{eq:scalar-zak-modulation-isomorphism}
Z_D:
M^p(\mathbb R^n)
\longrightarrow
W^{\infty,p}_{\mathrm{qp},\Gamma_D,\Gamma^\circ_D }(\mathbb R^{2n}),
\qquad p\in\{1,\infty\}.
\end{equation}

The scalar Zak transform also preserves the natural modulation-space
duality. By \cite[Theorem~3.5(2)]{ToftZak},
\begin{equation}
\label{eq:scalar-zak-duality}
\langle\psi,h\rangle_{M^\infty,M^1}
=
\langle Z_D\psi,Z_Dh\rangle_{\mathrm{qp},\Gamma_D,\Gamma^\circ_D },
\qquad
\psi\in M^\infty(\mathbb R^n),\quad
h\in M^1(\mathbb R^n).
\end{equation}
where 
\begin{equation*}
\label{eq:scalar-zak-pairing}
\langle F,G\rangle_{\mathrm{qp},\Gamma_D,\Gamma^\circ_D }
:=
\int_{Q_D\times Q_D^\circ}
F(x,\omega)\overline{G(x,\omega)}
\,dx\,d\omega.
\end{equation*}

% ============================================================
\subsection{The vector-valued Zak transform and the logical fibre}
\label{subsec:zak-logical-fibre}
% ============================================================

Let
\[
\mathbb Z_{\mathbf d}
:=
\mathbb Z_{d_1}\times\cdots\times\mathbb Z_{d_n},
\qquad
\mathcal H_{\mathbf d}
:=
\bigotimes_{j=1}^n\mathbb C^{d_j}
\cong
\mathbb C^{d},
\]
with canonical basis
\(\{e_r:r\in\mathbb Z_{\mathbf d}\}\). We define the
\emph{syndrome torus} by
\[
\mathbb T_D^{\mathrm{syn}}
:=
\mathbb R^{2n}/\Lambda_D^\circ.
\]
Since $\Lambda_D^\circ
=
\Gamma_D^\circ\times\Gamma_D^\circ$,
the set \(Q_D^\circ\times Q_D^\circ\) is a fundamental domain for
\(\mathbb T_D^{\mathrm{syn}}\). For \(r\in\mathbb Z_{\mathbf d}\), set $\delta_r:=D^{-1/2}r$.

Since $[\Gamma_D^\circ:\Gamma_D]=d$,
the points \(\delta_r\), \(r\in\mathbb Z_{\mathbf d}\), form a complete
set of representatives for
\(\Gamma_D^\circ/\Gamma_D\). Consequently, up to boundaries,
\begin{equation}
\label{eq:fundamental-domain-decomposition}
Q_D
=
\bigsqcup_{r\in\mathbb Z_{\mathbf d}}
(\delta_r+Q_D^\circ),
\end{equation}
and hence
\[
Q_D\times Q_D^\circ
=
\bigsqcup_{r\in\mathbb Z_{\mathbf d}}
\bigl((\delta_r+Q_D^\circ)\times Q_D^\circ\bigr).
\]

This decomposition collects the values of the scalar Zak transform
on the \(d\) sheets into a single
\(\mathcal H_{\mathbf d}\)-valued section. For
\(f\in\mathcal S(\mathbb R^n)\), define
\begin{equation}
\label{eq:vector-valued-zak}
\bigl(\mathcal Z_Df(x,\omega)\bigr)_r
:=
e^{-2\pi i\delta_r\cdot\omega}
Z_Df(x+\delta_r,\omega),
\qquad r\in\mathbb Z_{\mathbf d}.
\end{equation}
The decomposition
\eqref{eq:fundamental-domain-decomposition} and the unitarity of
\(Z_D\) show that \(\mathcal Z_D\) extends uniquely to a unitary
operator
\[
\mathcal Z_D:
L^2(\mathbb R^n)
\longrightarrow
L^2\bigl(
Q_D^\circ\times Q_D^\circ;
\mathcal H_{\mathbf d}
\bigr).
\]
Let \(p\in\{1,\infty\}\), and fix a nonzero window
\(\Phi\in\mathcal S(\mathbb R^{2n})\). Let
\[
W^{\infty,p}_{\mathrm{qp},\Gamma^\circ_D,\Gamma^\circ_D }
   (\mathbb R^{2n};\mathcal H_{\mathbf d})
\]
be the space of all
\(F\in \mathcal S'(\mathbb R^{2n};\mathcal H_{\mathbf d})\)
satisfying
\begin{align}
F_r(x+\delta_t,\omega)
&=
e^{2\pi i\delta_t\cdot\omega}
F_{r+t}(x,\omega),
\label{eq:vector-qp-x}
\\
F_r(x,\omega+\delta_s)
&=
e^{-2\pi i r^TD^{-1}s}
F_r(x,\omega),
\label{eq:vector-qp-w}
\end{align}
for \(r,t,s\in\mathbb Z_{\mathbf d}\), distributionally, and such
that
\[
\|F\|_{W^{\infty,p}_{\mathrm{qp},\Gamma^\circ_D,\Gamma^\circ_D }}
:=
\|\|(V_\Phi F_r)_{r\in \mathbb Z_{\mathbf d}}\|_{\mathcal H_{\mathbf d}}\|_
 {L^p((Q^\circ_D\times Q^\circ_D)\times\mathbb R^{2n})}
<\infty.
\]
For \(p=\infty\), the \(L^\infty\)-norm has its usual
essential-supremum interpretation.

\begin{lemma}
\label{lem:zak-mild}
The vector-valued Zak transform extends uniquely to topological
isomorphisms
\[
\mathcal Z_D:
M^p(\mathbb R^n)
\longrightarrow
W^{\infty,p}_{\mathrm{qp},\Gamma^\circ_D,\Gamma^\circ_D }
   (\mathbb R^{2n};\mathcal H_{\mathbf d}),
\qquad p\in\{1,\infty\}.
\]
For \(p=\infty\), the isomorphism and its inverse are weak-$*$
continuous with respect to the dualities with \(p=1\).

Moreover,
\begin{equation}
\label{eq:vector-zak-duality}
\langle\psi,h\rangle_{M^\infty,M^1}
=
\langle\mathcal Z_D\psi,\mathcal Z_Dh
\rangle_{\mathrm{qp},\Gamma^\circ_D,\Gamma^\circ_D }
\end{equation}
for every
\(\psi\in M^\infty(\mathbb R^n)\) and
\(h\in M^1(\mathbb R^n)\), where the pairing on the right is the
continuous extension of
\[
\langle F,G\rangle_{\mathrm{qp},\Gamma^\circ_D,\Gamma^\circ_D }
:=
\int_{Q^\circ_D\times Q^\circ_D }
\langle F(x,\omega),G(x,\omega)
\rangle_{\mathcal H_{\mathbf d}}
\,dx\,d\omega.
\]
\end{lemma}
\begin{proof}
Define the finite-sheet unfolding operator by
\[
\mathscr U_DG(x,\omega)
:=
\left(
e^{-2\pi i\delta_r\cdot\omega}
G(x+\delta_r,\omega)
\right)_{r\in\mathbb Z_{\mathbf d}}.
\]
Then $\mathcal Z_D=\mathscr U_DZ_D$.

Then \eqref{eq:fundamental-domain-decomposition} 
shows that \(\mathscr U_D\) replaces a scalar function on
\(Q_D\times Q_D^\circ\) by a vector-valued function on the smaller
domain $Q_D^\circ\times Q_D^\circ$.
Since only finitely many translations and modulations are involved,
\(\mathscr U_D\) is a topological isomorphism from $W^{\infty,p}_{\mathrm{qp},\Gamma_D,\Gamma^\circ_D }(\mathbb R^{2n})$ onto $W^{\infty,p}_{\mathrm{qp},\Gamma^\circ_D,\Gamma^\circ_D }
   (\mathbb R^{2n};\mathcal H_{\mathbf d})$.
The scalar Zak transition relations
\eqref{eq:wie-qpx}--\eqref{eq:wie-qpw} become precisely
\eqref{eq:vector-qp-x}--\eqref{eq:vector-qp-w} under
\(\mathscr U_D\).

Conversely, if
\(F\in
W^{\infty,p}_{\mathrm{qp},\Gamma^\circ_D,\Gamma^\circ_D }
(\mathbb R^{2n};\mathcal H_{\mathbf d})\),
define \(G\) on \(Q_D\times Q_D^\circ\) by
\[
G(x+\delta_r,\omega)
:=
e^{2\pi i\delta_r\cdot\omega}F_r(x,\omega),
\qquad
x\in Q_D^\circ.
\]
The vector transition relations ensure that this definition is
consistent on the boundaries and extends uniquely to a scalar
quasiperiodic distribution on \(\mathbb R^{2n}\). Thus
\(\mathscr U_D^{-1}F=G\).

Combining this finite-sheet isomorphism with the scalar Zak
isomorphism gives
\[
\mathcal Z_D:
M^p(\mathbb R^n)
\longrightarrow
W^{\infty,p}_{\mathrm{qp},\Gamma^\circ_D,\Gamma^\circ_D }
   (\mathbb R^{2n};\mathcal H_{\mathbf d}).
\]
The weak-$*$ assertion follows from the corresponding scalar result
and the weak-$*$ continuity of the finite-sheet unfolding.

For \(\psi,h\in\mathcal S(\mathbb R^n)\), scalar Zak duality and the
finite decomposition of \(Q_D\) give
\begin{align*}
\langle\psi,h\rangle
&=
\sum_{r\in\mathbb Z_{\mathbf d}}
\int_{Q_D\times Q_D^\circ}
Z_D\psi(x+\delta_r,\omega)
\overline{Z_Dh(x+\delta_r,\omega)}
\,dx\,d\omega
\\
&=
\int_{Q_D\times Q_D^\circ}
\left\langle
\mathcal Z_D\psi(x,\omega),
\mathcal Z_Dh(x,\omega)
\right\rangle_{\mathcal H_{\mathbf d}}
\,dx\,d\omega.
\end{align*}
The scalar
\(M^\infty\)--\(M^1\) Zak duality and the continuity of
\(\mathscr U_D\) extend the identity to
\(\psi\in M^\infty(\mathbb R^n)\) and
\(h\in M^1(\mathbb R^n)\).
\end{proof}

The vector-valued Zak transform diagonalizes the action of the
canonical stabilizers.

\begin{lemma}
\label{lem:stabilizer-zak-covariance}
Let \(f\in M^\infty(\mathbb R^n)\). Then
\begin{equation}
\label{eq:stabilizer-zak-character}
\mathcal Z_D(S_\lambda f)(z)
=
e^{-2\pi i\sigma(\lambda,z)}
\mathcal Z_Df(z),
\qquad
\lambda\in\Lambda_D.
\end{equation}
\end{lemma}

\begin{proof}
Write
\[
\lambda=(a,b)
=
(\sqrt D\,k,\sqrt D\,\ell)
\in\Lambda_D,
\qquad
k,\ell\in\mathbb Z^n.
\]
We first prove the formula for \(f\in\mathcal S(\mathbb R^n)\).
Using 
$\rho(a,b)
=
e^{-\pi i a\cdot b}M_bT_a$,
we obtain
\[
\begin{aligned}
Z_D(\rho(a,b)f)(x,\omega)
&=
\sum_{\gamma\in\Gamma_D}
(\rho(a,b)f)(x-\gamma)
e^{2\pi i\gamma\cdot\omega}
\\
&=
e^{-\pi i a\cdot b}
\sum_{\gamma\in\Gamma_D}
e^{2\pi i b\cdot(x-\gamma)}
f(x-\gamma-a)
e^{2\pi i\gamma\cdot\omega}.
\end{aligned}
\]
Since \ $b\cdot\gamma\in\mathbb Z$, we have that 
\[
Z_D(\rho(a,b)f)(x,\omega)
=
e^{-\pi i a\cdot b}
e^{2\pi ib\cdot x}
\sum_{\gamma\in\Gamma_D}
f(x-a-\gamma)e^{2\pi i\gamma\cdot\omega}.
\]
The change of variables \(\eta=\gamma+a\) gives
\[
\sum_{\gamma\in\Gamma_D}
f(x-a-\gamma)e^{2\pi i\gamma\cdot\omega}
=
e^{-2\pi ia\cdot\omega}Z_Df(x,\omega).
\]
Therefore,
\begin{equation*}
\label{eq:scalar-zak-weyl-action}
Z_D(\rho(\lambda)f)(z)
=
e^{-\pi ia\cdot b}
e^{-2\pi i\sigma(\lambda,z)}
Z_Df(z).
\end{equation*}
Since \(a\cdot b=k^{\mathsf T}D\ell\), we have
\(e^{-\pi i a\cdot b}=\chi_D(\lambda)\).
Using \(S_\lambda=\chi_D(\lambda)\rho(\lambda)\) and
\(\chi_D(\lambda)^2=1\), the phase factors cancel, giving
\[
Z_D(S_\lambda f)(z)
=
e^{-2\pi i\sigma(\lambda,z)}Z_Df(z).
\]
We now pass to the vector-valued transform.  By \eqref{eq:vector-valued-zak}
\[
\begin{aligned}
(\mathcal Z_D(S_\lambda f)(x,\omega))_r
&=
e^{-2\pi i\delta_r\cdot\omega}
Z_D(S_\lambda f)(x+\delta_r,\omega)
\\
&=
e^{-2\pi i\sigma(\lambda,(x+\delta_r,\omega))}
(\mathcal Z_Df(x,\omega))_r.
\end{aligned}
\]
Since $b\cdot\delta_r
=
(\sqrt D\,\ell)\cdot(D^{-1/2}r)
=
\ell\cdot r
\in\mathbb Z$,
we have
\[
e^{-2\pi i\sigma(\lambda,(x+\delta_r,\omega))}
=
e^{-2\pi i\sigma(\lambda,(x,\omega))}.
\]
This proves \eqref{eq:stabilizer-zak-character} for Schwartz functions.

Finally, Weyl operators act weak-$*$ continuously on
\(M^\infty(\mathbb R^n)\), and \(\mathcal Z_D\) is weak-$*$
continuous by Lemma~\ref{lem:zak-mild}.  The identity therefore extends
to every \(f\in M^\infty(\mathbb R^n)\).
\end{proof}

\subsection{The GKP codes}

The preceding results provide the functional-analytic setting in
which the ideal GKP code has exactly the expected finite logical
dimension. Lemma~\ref{lem:zak-mild} identifies
\(M^\infty(\mathbb R^n)\) with a space of vector-valued
quasiperiodic Zak distributions, while Lemma \ref{lem:stabilizer-zak-covariance} turns the stabilizer
conditions into multiplication equations on the syndrome torus.
These equations force the Zak transform of a stabilized state to be
supported at the trivial syndrome.

The choice of \(M^\infty(\mathbb R^n)\) as the ambient space is
important here. In the full space of tempered distributions, support
at the syndrome lattice would still permit derivatives of Dirac
masses and polynomially growing coefficient sequences. By contrast,
the \(M^\infty\) condition implies that the Zak transform is an
order-zero lattice Dirac comb with uniformly bounded coefficients.
The vector-valued transition relations determine all these
coefficients from the one at the origin. We may therefore identify
the GKP code bijectively with the finite-dimensional logical fibre
\(\mathcal H_{\mathbf d}\).

\begin{definition}[GKP code]\label{def:ideal-code}
For a symplectically integral lattice $\Lambda$ and a stabilizer phase $\chi$, we define its associated  \emph{GKP code}
\[
 \cC_{\Lambda,\chi}
 =\{\psi\in M^\infty(\R^n):S_\lambda\psi=\psi
 \text{ for every }\lambda\in\Lambda\}.
\]
Equivalently,
\[
\mathcal C_{\Lambda,\chi}
=
\{\psi\in M^\infty(\R^n):
\rho(\lambda)\psi=\overline{\chi(\lambda)}\psi
\text{ for every }\lambda\in\Lambda\}.
\]
\end{definition}

For the rectangular lattice with its canonical stabilizer phase, we
write
\[
\mathcal C_{\Lambda_D}
:=
\mathcal C_{\Lambda_D,\chi_D}.
\]

We can now characterize the ideal code intrinsically as the fibre over
the trivial syndrome.

\begin{proposition}[Zak characterization of the ideal code]
\label{prop:Zak-characterization}
For every \(\psi\in\mathcal C_{\Lambda_D}\), there exists a unique
vector \(c_\psi\in\mathcal H_{\mathbf d}\) such that
\[
\mathcal Z_D\psi=c_\psi\delta_0
\]
as a distributional section over $\mathbb T_D^{\mathrm{syn}}$. Moreover, the map
\[
\mathcal C_{\Lambda_D}
\longrightarrow
\mathcal H_{\mathbf d},
\qquad
\psi\longmapsto c_\psi,
\]
is a linear isomorphism.
\end{proposition}

\begin{proof}
Let \(\psi\in\mathcal C_{\Lambda_D}\). Since
\[
S_\lambda\psi=\psi,
\qquad \lambda\in\Lambda_D,
\]
Lemma~\ref{lem:stabilizer-zak-covariance} gives
\begin{equation*}
\label{eq:zak-stabilizer-annihilation}
\left(
e^{-2\pi i\sigma(\lambda,z)}-1
\right)
\mathcal Z_D\psi
=
0,
\qquad
\lambda\in\Lambda_D.
\end{equation*}
It follows that $\operatorname{supp}(\mathcal Z_D\psi)
\subseteq
\Lambda_D^\circ$.
By Lemma~\ref{lem:zak-mild}, $\mathcal Z_D\psi
\in W^{\infty,\infty}_{\mathrm{qp},\Gamma^\circ_D,\Gamma^\circ_D }
   (\mathbb R^{2n};\mathcal H_{\mathbf d})\subseteq 
M^\infty(\mathbb R^{2n};\mathcal H_{\mathbf d})$.
Since \(\Lambda_D^\circ\) is uniformly discrete,
Step~\textup{(II)} in the proof of
\cite[Theorem~7.7.5]{FeiKo98} implies, componentwise, that the lift of
\(\mathcal Z_D\psi\) to \(\mathbb R^{2n}\) is of the form
\begin{equation*}
\label{eq:lifted-zak-dirac-comb}
\mathcal Z_D\psi
=
\sum_{\nu\in\Lambda_D^\circ}
c_\nu\delta_\nu,
\qquad
(c_\nu)_{\nu\in\Lambda_D^\circ}
\in
\ell^\infty(
\Lambda_D^\circ;\mathcal H_{\mathbf d}
).
\end{equation*}
The quasi periodicity relations determine all coefficients \(c_\nu\) from
the coefficient \(c_0\) at the origin. 
Thus, as a distributional section on the syndrome torus,
\[
\mathcal Z_D\psi=c_\psi\delta_0,
\qquad
c_\psi:=c_0.
\]
Uniqueness of \(c_\psi\) is immediate, and injectivity of
\(\psi\mapsto c_\psi\) follows from injectivity of
\(\mathcal Z_D\).

Conversely, let \(c\in\mathcal H_{\mathbf d}\). The quasi periodicity
relations \eqref{eq:vector-qp-x} and
\eqref{eq:vector-qp-w} determine, for every
\(\nu\in\Lambda_D^\circ\), a unitary transition operator
\[
\tau_\nu:\mathcal H_{\mathbf d}\longrightarrow
\mathcal H_{\mathbf d}.
\]
The point-supported section \(c\delta_0\) therefore has the
quasiperiodic lift
\[
F_c:=\sum_{\nu\in\Lambda_D^\circ}\tau_\nu(c)\delta_\nu.
\]
The transition operators are unitary, and hence
\[
\|\tau_\nu(c)\|_{\mathcal H_{\mathbf d}}
=
\|c\|_{\mathcal H_{\mathbf d}},
\qquad \nu\in\Lambda_D^\circ.
\]
Since \(F_c\) is a lattice Dirac comb with uniformly bounded
coefficients, the rapid decay of the STFT window implies that
 $F_c\in W^{\infty,\infty}_{\mathrm{qp},\Gamma^\circ_D,\Gamma^\circ_D }
   (\mathbb R^{2n};\mathcal H_{\mathbf d})$. Then
Lemma~\ref{lem:zak-mild} gives a unique
\[
\psi_c
:=
\mathcal Z_D^{-1}(F_c)
\in M^\infty(\mathbb R^n).
\]

It remains to verify that \(\psi_c\) belongs to the code. Since
\(F_c=c\delta_0\) on the syndrome torus, multiplication by the
stabilizer character gives
\[
e^{-2\pi i\sigma(\lambda,z)}F_c
=
F_c,
\qquad \lambda\in\Lambda_D,
\]
because the character equals \(1\) at \(z=0\). By
Lemma~\ref{lem:stabilizer-zak-covariance},
\[
\mathcal Z_D(S_\lambda\psi_c)
=
e^{-2\pi i\sigma(\lambda,z)}
\mathcal Z_D\psi_c
=
\mathcal Z_D\psi_c.
\]
Injectivity of \(\mathcal Z_D\) therefore implies $S_\lambda\psi_c=\psi_c,
\qquad \lambda\in\Lambda_D$.
Hence
\[
\psi_c\in\mathcal C_{\Lambda_D},
\]
and the map \(\psi\mapsto c_\psi\) is surjective.
\end{proof}

\begin{definition}[Logical vector]
For \(\psi\in\mathcal C_{\Lambda_D}\), the unique vector
\(c_\psi\in\mathcal H_{\mathbf d}\) determined by
\[
\mathcal Z_D\psi=c_\psi\delta_0
\]
is called the \emph{logical vector} of \(\psi\).
\end{definition}

For \(c\in\mathcal H_{\mathbf d}\), we denote by \(\psi_c\) the unique
ideal codeword satisfying
\[
\mathcal Z_D\psi_c=c\delta_0
\]
as a distributional section over
\(\mathbb T_D^{\mathrm{syn}}\).

\subsection{Pauli operators and Clifford covariance}
\label{subsec:finite-weyl-background}

For \(i=1,\ldots,n\), define the shift and phase operators on
\(\cH_{\mathbf d}\) by
\[
X_i e_r=e_{r+\varepsilon_i},
\qquad
Z_i e_r=e^{2\pi i r_i/d_i}e_r
\]
where $\{e_r:r\in\Z_{\mathbf d}\}$ is the canonical basis of $\mathcal H_{\mathbf d}$.
The \emph{logical Pauli group} is the group
\[
\mathcal P_{\mathbf d}
:=
\langle X_1,\ldots,X_n,Z_1,\ldots,Z_n, e^{i\pi/d_n}I\rangle .
\]
The relation with phase-space displacements is particularly simple in
Zak coordinates. Let
\[
a_i=(D^{-1/2}\varepsilon_i,0),
\qquad
b_i=(0,D^{-1/2}\varepsilon_i)
\in\Lambda_D^\circ.
\]
\begin{lemma}
\label{lem:adjoint-generators-pauli}
For every \(h\in M^1(\mathbb R^n)\) and \(i=1,\ldots,n\),
\begin{equation}
\label{eq:adjoint-generators-pauli}
\cZ_D(\rho(a_i)h)(0)
=
X_i\,\cZ_Dh(0),
\qquad
\cZ_D(\rho(b_i)h)(0)
=
Z_i\,\cZ_Dh(0).
\end{equation}
Thus, the elementary generators of the adjoint lattice act on the
logical Zak fibre as the shift and phase operators.
\end{lemma}

\begin{proof}
Set \(c=\cZ_Dh(0)\). By \eqref{eq:vector-valued-zak},
\[
c_r=Z_Dh(\delta_r,0)
=\sqrt d\sum_{\gamma\in\Gamma_D}h(\delta_r-\gamma).
\]
Since \(\rho(a_i)=T_{D^{-1/2}\varepsilon_i}\), periodicity of
\(Z_Dh(\,\cdot\,,0)\) with respect to \(\Gamma_D\) gives
\[
\bigl(\cZ_D(\rho(a_i)h)(0)\bigr)_r
=
Z_Dh(\delta_r-D^{-1/2}\varepsilon_i,0)
=
c_{r-\varepsilon_i}
=
(X_i c)_r,
\]
where indices are taken in \(\mathbb Z_{\mathbf d}\).
Similarly, \(\rho(b_i)=M_{D^{-1/2}\varepsilon_i}\), and
\(D^{-1/2}\varepsilon_i\cdot\gamma\in\mathbb Z\) for
\(\gamma\in\Gamma_D\). Hence
\[
\bigl(\cZ_D(\rho(b_i)h)(0)\bigr)_r
=
e^{2\pi i r_i/d_i}c_r
=
(Z_i c)_r.
\]
The sums converge absolutely for \(h\in M^1(\mathbb R^n)\).
\end{proof}

We identify
\[
\Z_{\mathbf d}\times\Z_{\mathbf d}\cong K_D:=\Lambda_D^\circ/\Lambda_D
\]
via
\begin{equation} \label{identification_K}
 q=(r,s)
\longmapsto
\nu_q+\Lambda_D,
\qquad
\nu_q:=
(D^{-1/2}r,D^{-1/2}s)\in\Lambda_D^\circ.   
\end{equation}
For \(q=(r,s)\in K_D\), we then define the corresponding Pauli operator by
\begin{equation}\label{pauli_op}
\mathsf W_q
:=e^{\pi i r^TD^{-1}s}
X_1^{r_1}\cdots X_n^{r_n}
Z_1^{s_1}\cdots Z_n^{s_n}.    
\end{equation}
The Weyl relations \eqref{eq:weyl-commutation_0} and \eqref{eq:adjoint-generators-pauli} yield
\begin{equation*}
\cZ_D(\rho(\nu_q)h)(0)
=
\mathsf W_q\,\cZ_Dh(0).
\end{equation*}
for every $h\in M^1(\mathbb R^n)$.
Given $\psi_c\in \mathcal C_{\Lambda_D}$ by duality of the Zak transform
\[
\begin{aligned}
\langle \rho(\nu_q)\psi_c,h\rangle
&=\langle \psi_c,\rho(-\nu_q)h\rangle\\
&=\langle c\delta_0, \mathcal Z_D(\rho(-\nu_q)h)\rangle\\
&=\langle c, \mathcal Z_D(\rho(-\nu_q)h)(0)\rangle\\
&=\langle c,\mathsf W_q^*\mathcal Z_D(h)(0)\rangle.
\end{aligned}
\]
thus
\[\rho(\nu_q)\psi_c=\psi_{\mathsf W_q c}\,.\]

\subsection{Metaplectic operators and code spaces}\label{subsec:metaplectic-operators}

We briefly recall the basic properties of the metaplectic
representation that will be used below. For every
\(T\in\operatorname{Sp}(2n,\mathbb R)\), there exists a unitary
operator
\[
U_T:L^2(\mathbb R^n)\longrightarrow L^2(\mathbb R^n),
\]
unique up to a phase, such that
\begin{equation}
\label{eq:metaplectic-covariance}
U_T\rho(z)U_T^*
=
\rho(Tz),
\qquad z\in\mathbb R^{2n};
\end{equation}
see \cite{GjertsenLuef}. An important regularity property is that metaplectic operators act
continuously and bijectively on every modulation space
\(M^p(\mathbb R^n)\), \(1\leq p\leq\infty\). In particular, they are
weak-$*$ continuous automorphisms of \(M^\infty(\mathbb R^n)\).
They also act continuously on \(\mathcal S(\mathbb R^n)\) and, by
duality, on \(\mathcal S'(\mathbb R^n)\). Thus the covariance
relation \eqref{eq:metaplectic-covariance} remains valid throughout
the modulation-space scale.
Moreover,  \(U_T\) maps \(\cG(\mathbf g,\Lambda)\) onto
\(\cG(U_T\mathbf g,T\Lambda)\). It therefore preserves the frame, Riesz-sequence,
and orthonormal-basis properties, together with their bounds. In
particular, if \(\boldsymbol\gamma\) is dual to \(\mathbf g\) over \(\Lambda\), then
\(U_T\boldsymbol\gamma\) is dual to \(U_T\mathbf g\) over \(T\Lambda\); see
\cite{GjertsenLuef}.

Let \(\Lambda\subset\mathbb R^{2n}\) be a symplectically integral
lattice, and let $\chi$ be a stabilizer phase.  The symplectic
transformation \(T\) transports \(\chi\) to a stabilizer phase on
\(T\Lambda\), defined by
\begin{equation*}
\chi^T(\mu)
:=
\chi(T^{-1}\mu),
\qquad
\mu\in T\Lambda.
\end{equation*}

\begin{proposition}
\label{prop:metaplectic-code-transport}
Let \(U_T\) be a metaplectic lift of \(T\). Then
\[
U_T\mathcal C_{\Lambda,\chi}
=
\mathcal C_{T\Lambda,\chi^T}.
\]
\end{proposition}

\begin{proof}
Let \(\psi\in\mathcal C_{\Lambda,\chi}\) and write
\(\mu=T\lambda\in T\Lambda\). By metaplectic covariance,
\begin{align*}
S_\mu^{\chi^T}U_T\psi
&=
\chi^T(\mu)\rho(\mu)U_T\psi
=
\chi(\lambda)\rho(T\lambda)U_T\psi =
U_T\chi(\lambda)\rho(\lambda)\psi
=
U_TS_\lambda^\chi\psi
=
U_T\psi.
\end{align*}
Thus, $U_T\psi\in\mathcal C_{T\Lambda,\chi^T}$.
Applying the same argument to \(T^{-1}\) gives the reverse inclusion.
\end{proof}

\begin{remark}
If \(T\Lambda=\Lambda\), then \(U_T\) preserves the code space $\mathcal C_{\Lambda,\chi}$ precisely when $\chi\circ T^{-1}=\chi$.
\end{remark}

\subsubsection{Clifford gates and metaplectic transformations}
\label{Clifford_gates}

We will see that metaplectic unitaries preserving the code space induce Clifford gates on the logical information.

\begin{lemma}
Let $T\in\operatorname{Sp}(2n,\mathbb R)$ satisfy
\[
T\Lambda_D=\Lambda_D,
\qquad
\chi_D\circ T^{-1}=\chi_D,
\]
and let $U_T$ be a metaplectic lift. Then the induced map
\[
H_T:\mathcal H_{\mathbf d}\longrightarrow\mathcal H_{\mathbf d},
\qquad
\mathcal Z_D(U_T\psi_c)=H_Tc\,\delta_0,
\]
is unitary.
\end{lemma}

\begin{proof}
Metaplectic covariance and preservation of the stabilizer phase imply
that \(U_T\) preserves \(\mathcal C_{\Lambda_D}\); see
Proposition~\ref{prop:metaplectic-code-transport}.  

Define the unitary operator
\[V_T:=\mathcal Z_D U_T \mathcal Z_D^{-1}:L^2\bigl(
\mathbb T^{\textrm{syn}}_D;
\mathcal H_{\mathbf d})\to L^2\bigl(
\mathbb T^{\textrm{syn}}_D;
\mathcal H_{\mathbf d})\,.\]
For \(\lambda\in\Lambda_D\), set $\chi_\lambda(z):=e^{-2\pi i\sigma(\lambda,z)}$.
By Lemma~\ref{lem:stabilizer-zak-covariance},
\[
\mathcal Z_DS_\lambda\mathcal Z_D^{-1}
=
M_{\chi_\lambda}.
\]
Since \(U_TS_\lambda U_T^*=S_{T\lambda}\), we obtain, for every
\(F\in
L^2(\mathbb T_D^{\mathrm{syn}};\mathcal H_{\mathbf d})\),
\begin{align*}
V_TM_{\chi_\lambda}F
&=
\mathcal Z_DU_TS_\lambda\mathcal Z_D^{-1}F
\\
&=
\mathcal Z_DS_{T\lambda}U_T\mathcal Z_D^{-1}F
\\
&=
M_{\chi_{T\lambda}}V_TF.
\end{align*}
Because \(T\) is symplectic,
\[
\chi_{T\lambda}(z)
=
e^{-2\pi i\sigma(T\lambda,z)}
=
e^{-2\pi i\sigma(\lambda,T^{-1}z)}
=
\chi_\lambda(T^{-1}z).
\]
Therefore,
\[
V_TM_{\chi_\lambda}
=
M_{\chi_\lambda\circ T^{-1}}V_T,
\qquad \lambda\in\Lambda_D.
\]

The characters
\(\{\chi_\lambda:\lambda\in\Lambda_D\}\) form the Pontryagin dual of
\(\mathbb T_D^{\mathrm{syn}}\). By Fourier approximation, the
preceding identity extends first to every
\(f\in C(\mathbb T_D^{\mathrm{syn}})\), and then, by weak-$*$
density and normality, to every
\(f\in L^\infty(\mathbb T_D^{\mathrm{syn}})\):
\begin{equation}
\label{eq:VT-multiplier-covariance}
V_TM_f
=
M_{f\circ T^{-1}}V_T.
\end{equation}

Let $(A_TF)(z):=F(T^{-1}z)$. Since \(T\) preserves Haar measure on
\(\mathbb T_D^{\mathrm{syn}}\), \(A_T\) is unitary, and
\[
A_TM_f=M_{f\circ T^{-1}}A_T.
\]
It follows from
\eqref{eq:VT-multiplier-covariance} that
\[
W_T:=V_TA_T^*
\]
commutes with every multiplication operator \(M_f\),
\(f\in L^\infty(\mathbb T_D^{\mathrm{syn}})\).

By the characterization of decomposable operators as the commutant
of the diagonal algebra
\cite[Theorem~14.1.10]{KadisonRingroseII}, there exists an
essentially bounded measurable field
\[
R_T:
\mathbb T_D^{\mathrm{syn}}
\longrightarrow
B(\mathcal H_{\mathbf d})
\]
such that
\[
(W_TF)(z)=R_T(z)F(z)
\]
for almost every \(z\). Since \(W_T\) is unitary,
\[
R_T(z)\in\mathcal U(\mathcal H_{\mathbf d})
\]
for almost every \(z\). Consequently,
\begin{equation}
\label{eq:VT-weighted-composition}
(V_TF)(z)
=
R_T(z)F(T^{-1}z)
\end{equation}
for almost every \(z\). By standard arguments  one can show that \(R_T\) has a continuous representative in a
neighbourhood of the origin.  In particular, \(R_T(0)\) is well defined and unitary.

Let \(c\in\mathcal H_{\mathbf d}\) and
\(h\in M^1(\mathbb R^n)\). By the vector-valued Zak duality  \eqref{eq:vector-zak-duality},
\begin{align*}
\langle U_T\psi_c,h\rangle
&=
\langle\psi_c,U_T^*h\rangle
\\
&=
\left\langle
\mathcal Z_D\psi_c,
\mathcal Z_D(U_T^*h)
\right\rangle
\\
&=
\left\langle
c\,\delta_0,
V_{T^{-1}}(\mathcal Z_Dh)
\right\rangle
\\
&=
\left\langle
c,
R_{T^{-1}}(0)\mathcal Z_Dh(0)
\right\rangle
\\
&=
\left\langle
R_{T^{-1}}(0)^*c,
\mathcal Z_Dh(0)
\right\rangle
\\
&=
\left\langle
\psi_{R_{T^{-1}}(0)^*c},
h
\right\rangle.
\end{align*}
Since this holds for every \(h\in M^1(\mathbb R^n)\), we conclude
that
\[
U_T\psi_c
=
\psi_{R_{T^{-1}}(0)^*c}.
\]
Equivalently,
\[
\mathcal Z_D(U_T\psi_c)
=
R_{T^{-1}}(0)^*c\,\delta_0.
\]
Therefore, the induced logical operator is $H_T=R_{T^{-1}}(0)^*$.
Since \(R_{T^{-1}}(0)\) is unitary, \(H_T\) is unitary.
\end{proof}

We now compare this action with the logical Pauli operators. Recall that
\[
    a_i=(D^{-1/2}\varepsilon_i,0),
    \qquad
    b_i=(0,D^{-1/2}\varepsilon_i),
\]
act on the logical fibre as \(X_i\) and \(Z_i\), respectively. Since
\(T^{-1}a_i,T^{-1}b_i\in\Lambda_D^\circ\), define
\[
    q_i=[T^{-1}a_i],
    \qquad
    p_i=[T^{-1}b_i]
    \qquad\text{in }K_D,
\]
and choose the corresponding representatives so that
\[
    T^{-1}a_i=\nu_{q_i}+\theta_i,
    \qquad
    T^{-1}b_i=\nu_{p_i}+\phi_i,
    \qquad
    \theta_i,\phi_i\in\Lambda_D.
\]
Using  \eqref{eq:weyl-commutation_0} together with the stabilizer relation we obtain
\begin{align*}
    \rho(T^{-1}a_i)\psi_c
    &=
    e^{\pi i\sigma(\nu_{q_i},\theta_i)}
    \overline{\chi_D(\theta_i)}
    \rho(\nu_{q_i})\psi_c,\\
    \rho(T^{-1}b_i)\psi_c
    &=
    e^{\pi i\sigma(\nu_{p_i},\phi_i)}
    \overline{\chi_D(\phi_i)}
    \rho(\nu_{p_i})\psi_c.
\end{align*}
Hence
\[
   c_{\rho(T^{-1}a_i)\psi_c}
    =
    \zeta_i\,\mathsf W_{q_i}c,
    \qquad
    c_{\rho(T^{-1}b_i)\psi_c}
    =
    \xi_i\,\mathsf W_{p_i}c,
\]
where
\[
    \zeta_i
    :=
    e^{\pi i\sigma(\nu_{q_i},\theta_i)}
    \overline{\chi_D(\theta_i)},
    \qquad
    \xi_i
    :=
    e^{\pi i\sigma(\nu_{p_i},\phi_i)}
    \overline{\chi_D(\phi_i)}.
\]
Using \eqref{eq:metaplectic-covariance} 
we have that 
\begin{align*}
    X_iH_Tc
    &=
    c_{\rho(a_i)U_T\psi_c}
     =
    c_{U_T\rho(T^{-1}a_i)\psi_c}
     =
    \zeta_i H_T\mathsf W_{q_i}c,\\
    Z_iH_Tc
    &=
    c_{\rho(b_i)U_T\psi_c}
     =
    c_{U_T\rho(T^{-1}b_i)\psi_c}
     =
    \xi_i H_T\mathsf W_{p_i}c.
\end{align*}
Equivalently,
\begin{equation*}\label{phases_clifford}
H_T^*X_iH_T
    =
    \zeta_i\,\mathsf W_{q_i},
    \qquad
    H_T^*Z_iH_T
    =
    \xi_i\,\mathsf W_{p_i}.    
\end{equation*}
Since \(\sigma(\Lambda_D^\circ,\Lambda_D)\subseteq\mathbb Z\)
and \(\chi_D\) takes values in \(\{\pm1\}\), the phases
\(\zeta_i\) and \(\xi_i\) belong to \(\{\pm1\}\).
Consequently, conjugation by \(H_T\) maps the generators of
\(\mathcal P_{\mathbf d}\) into \(\mathcal P_{\mathbf d}\).
Applying the same argument to \(T^{-1}\) gives
\[
H_T\mathcal P_{\mathbf d}H_T^*
=
\mathcal P_{\mathbf d}.
\]
Thus \(H_T\) belongs to the logical Clifford group
\[
\operatorname{Cliff}(\mathbf d)
:=
\bigl\{
V\in\mathcal U(\mathcal H_{\mathbf d})
:
V\mathcal P_{\mathbf d}V^*=\mathcal P_{\mathbf d}
\bigr\}.
\]
In particular, the physical metaplectic unitary \(U_T\)
implements the logical Clifford gate \(H_T\):
\[
U_T\psi_c=\psi_{H_Tc},
\qquad c\in\mathcal H_{\mathbf d}.
\]

\subsubsection{Example: The square GKP qubit and the Hadamard gate}

We now consider the single-mode square GKP lattice
\[
\Lambda_2
=
\sqrt2\,\mathbb Z\times\sqrt2\,\mathbb Z
 \qquad \text{with} \qquad \Lambda_2^\circ
=
\frac1{\sqrt2}\mathbb Z
\times
\frac1{\sqrt2}\mathbb Z.
\] 
The corresponding logical Hilbert space is $\mathcal H_{\mathrm{log}}\cong\mathbb C^2$, and the  canonical stabilizer phase is trivial. The logical Pauli generators are represented by the adjoint-lattice
vectors
\[
\lambda_X
=
\left(\frac1{\sqrt2},0\right),
\qquad
\lambda_Z
=
\left(0,\frac1{\sqrt2}\right).
\]
Under the Zak-transform identification, they act as
\[
X=
\begin{pmatrix}
0&1\\
1&0
\end{pmatrix},
\qquad
Z=
\begin{pmatrix}
1&0\\
0&-1
\end{pmatrix}.
\]
Consider the symplectic rotation
\[
J=
\begin{pmatrix}
0&1\\
-1&0
\end{pmatrix}.
\]
It preserves the square lattice $
J\Lambda_2=\Lambda_2$,
and trivially the stabilizer phase. Moreover,
\[
J\lambda_X=-\lambda_Z,
\qquad
J\lambda_Z=\lambda_X.
\]
The Fourier-transform  $\mathcal F$ is a metaplectic operator associated with \(J\), that is,
\[
\mathcal F\rho(z)\mathcal F^*
=
\rho(Jz).
\]
Consequently, the logical unitary \(H_J\) induced by \(\mathcal F\)
satisfies
$$
H_JXH_J^*=
Z^{-1}=Z,\qquad \text{and}\qquad
H_JZH_J^*=X.$$
The unitary exchanging \(X\) and \(Z\) is, up to a global phase, the
Hadamard gate
\[
H
=
\frac1{\sqrt2}
\begin{pmatrix}
1&1\\
1&-1
\end{pmatrix}.
\]
Therefore, under the identification $\mathcal C_{\Lambda_2}\cong\mathbb C^2$,
the Fourier transform acts as the logical Hadamard gate.
Equivalently, if $\mathcal Z_2\psi_c=c\,\delta_0$ for $c\in\mathbb C^2$, then
\[
\mathcal Z_2(\mathcal F\psi_c)
=Hc\,\delta_0.
\]
Thus, the quarter-turn symplectic symmetry of the square GKP lattice is
realized physically by the Fourier transform and acts on the encoded
qubit as the Hadamard gate.

\subsection{Displacement errors and the syndrome torus}

Let $0\ne\psi\in\mathcal C_{\Lambda_D}$, and let $\varepsilon=(\varepsilon_x,\varepsilon_\omega)
\in\mathbb R^{2n}$ be an arbitrary displacement error. The Zak covariance relation
shows that $\rho(\varepsilon)$ translates the support of the Zak
transform from the origin to the class of $\varepsilon$ in
$\mathbb T_D^{\mathrm{syn}}$, together with a phase and a corresponding
transformation of the fibre vector. Viewed as a distributional
section over the syndrome torus, we have
\[
\operatorname{supp}
\mathcal Z_D\bigl(\rho(\varepsilon)\psi\bigr)
=
\{[\varepsilon]\}.
\]
Thus, the base point
\[
[\varepsilon]\in
\mathbb T_D^{\mathrm{syn}}
=
\mathbb R^{2n}/\Lambda_D^\circ
\]
records the displacement syndrome, whereas the vector in the fibre
carries the logical information.

The displaced space is itself a GKP code for the same lattice,
with a possibly different stabilizer phase. Indeed, set
\[
\chi_D^\varepsilon(\lambda)
:=
e^{2\pi i\sigma(\lambda,\varepsilon)}\chi_D(\lambda),
\qquad \lambda\in\Lambda_D.
\]
The Weyl commutation relations imply
\[
\rho(\varepsilon)\mathcal C_{\Lambda_D,\chi_D}
=
\mathcal C_{\Lambda_D,\chi_D^\varepsilon}.
\]
The identification of the canonical code with the logical fibre
at the origin depends on the specific phase $\chi_D$.
For the modified stabilizers
\[
S_\lambda^\varepsilon
:=
\chi_D^\varepsilon(\lambda)\rho(\lambda),
\]
Lemma~4.2 gives, for $f\in M^\infty(\mathbb R^n)$,
\[
\mathcal Z_D(S_\lambda^\varepsilon f)(z)
=
e^{-2\pi i\sigma(\lambda,z-\varepsilon)}
\mathcal Z_Df(z).
\]
Consequently, nonzero elements of their joint $+1$ eigenspace
have Zak support at $[\varepsilon]$. Thus, for the fixed Zak
transform $\mathcal Z_D$, the stabilizer phase selects the syndrome
fibre representing the code.

In particular, $\chi_D^\varepsilon=\chi_D$ precisely when
$\varepsilon\in\Lambda_D^\circ$. Such displacements preserve
the canonical code space and act on its logical fibre.
For a general $\varepsilon$, applying $\rho(-\varepsilon)$
restores the canonical phase and returns the support to the origin.

\subsection{Construction of ideal codes from multi-window Gabor frames}
\label{Zak_matrix}
Multi-window Gabor frames provide explicit ideal codewords.

Let
$\mathbf g=(g_1,\ldots,g_m)
\in M^1(\mathbb R^n)\otimes\mathbb C^m$
be such that $\mathcal G(\mathbf g,\Lambda_D)$ 
is a multi-window Gabor frame. Let
\[
G(x,\omega)
=
\bigl(
\mathcal Z_Dg_1(x,\omega),\ldots,
\mathcal Z_Dg_m(x,\omega)
\bigr)
\]
be the associated \(d\times m\) Zak matrix.

The standard fibreization criterion for rational Gabor systems implies
that \(\mathcal G(\mathbf g,\Lambda_D)\) is a frame if and only if the
matrices \(G(x,\omega)\) have full row rank with uniform upper and lower
singular-value bounds almost everywhere; see
\cite[Chapter~8]{GrochenigBook}.  In the present setting the Zak
matrix is continuous, so the lower bound extends from almost every point
to every point of the compact base torus. Consequently, there exists \(A>0\) such that
\[
G(x,\omega)G(x,\omega)^*
\geq A I_{\mathcal H_{\mathbf d}}
\]
for every \((x,\omega)\in\T_D^{\mathrm{syn}}\). In particular,
\[
G(0,0):
\mathbb C^m\longrightarrow\mathcal H_{\mathbf d}
\]
is surjective.

Given \(\alpha=(\alpha_1,\ldots,\alpha_m)\in\mathbb C^m\), define
\[
g_\alpha
:=
\sum_{j=1}^m\alpha_jg_j
\in M^1(\mathbb R^n)
\]
and consider the stabilizer orbit sum
\begin{equation}
\label{eq:gabor-codeword}
\Psi_{\alpha,\mathbf g}
:=
\sum_{\lambda\in\Lambda_D}
S_\lambda g_\alpha
=
\sum_{\lambda\in\Lambda_D}
\chi_D(\lambda)\rho(\lambda)g_\alpha.
\end{equation}
The series in \eqref{eq:gabor-codeword} is understood in the weak-*
topology of $M^\infty(\mathbb R^n)$.
Indeed, for every \(h\in M^1(\mathbb R^n)\),
\[
\sum_{\lambda\in\Lambda_D}
\left|
\left\langle
\rho(\lambda)g_\alpha,h
\right\rangle
\right|
<\infty,
\]
by \eqref{weak_conv}.
Since \(\lambda\mapsto S_\lambda\) is a representation of
\(\Lambda_D\), reindexing the series gives
\[
S_\mu\Psi_{\alpha,\mathbf g}=\Psi_{\alpha,\mathbf g},
\qquad
\mu\in\Lambda_D.
\]
Hence $\Psi_{\alpha,\mathbf g}\in\mathcal C_{\Lambda_D}$. Now, using \eqref{eq:stabilizer-zak-character}, we obtain
\begin{align*}
\mathcal Z_D\Psi_{\alpha,\mathbf g}(z)
&=
\sum_{\lambda\in\Lambda_D}
e^{-2\pi i\sigma(\lambda,z)}
\mathcal Z_Dg_\alpha(z)
\nonumber\\
&=
\left(
\sum_{\lambda\in\Lambda_D}
e^{-2\pi i\sigma(\lambda,z)}
\right)
\mathcal Z_Dg_\alpha(z).
\end{align*}
The Pontryagin dual of $\T_D^{\mathrm{syn}}
=
\mathbb R^{2n}/\Lambda_D^\circ$
is naturally identified with \(\Lambda_D\) through the characters
\[
\chi_\lambda(z)
=
e^{-2\pi i\sigma(\lambda,z)},
\qquad
\lambda\in\Lambda_D.
\]
With Lebesgue measure in the
$x$-variable on $Q_D^\circ$ and the normalized frequency measure, the total mass of the
vector-valued Zak base is $\mu_D(\T_D^{\mathrm{syn}})=d^{-1/2}$.
For the Dirac distribution defined by evaluation at the identity, Fourier
inversion therefore gives
\[
\sum_{\lambda\in\Lambda_D}
e^{-2\pi i\sigma(\lambda,z)}
=
d^{-1/2}\delta_0(z)
\]
in the distributional sense. Consequently,
\begin{align*}
\mathcal Z_D\Psi_{\alpha,\mathbf g}(z)
&=
d^{-1/2}\delta_0(z)\,\mathcal Z_Dg_\alpha(z)
\nonumber\\
&=
d^{-1/2}\mathcal Z_Dg_\alpha(0,0)\,\delta_0(z)
\nonumber\\
&=
d^{-1/2}G(0,0)\alpha\,\delta_0(z).
\end{align*}
Since \(G(0,0)\) is surjective, for every $c\in\mathcal H_{\mathbf d}$ 
there exists \(\alpha\in\mathbb C^m\) such that $G(0,0)\alpha=c$.
Then the  corresponding distribution $\Psi_{\alpha,\mathbf g}$  satisfies
\[
\mathcal Z_D\Psi_{\alpha,\mathbf g}=d^{-1/2}c\,\delta_0.
\]
Rescaling $\alpha$ by $d^{1/2}$ produces the prescribed logical
vector $c$.

\subsection{Zak transform for symplectically integral lattices}
\label{subsec:general-zak}

We next explain how the logical Zak description for the rectangular
normal form $\Lambda_D$ induces a logical description for an arbitrary
symplectically integral lattice, and how the corresponding logical
displacement operators are identified.

By Proposition~\ref{prop:normal-form}, choose $T\in\operatorname{Sp}(2n,\mathbb R)$ with  $\Lambda=T\Lambda_D$,
and let $U_T$ be a metaplectic lift of $T$.  Applying Proposition~\ref{prop:metaplectic-code-transport} to $T^{-1}$
gives
\[
    U_T^{*}\mathcal C_{\Lambda,\chi}
    =
    \mathcal C_{\Lambda_D,\chi^T_D}.
\]
Thus, every ideal codeword for the general lattice $(\Lambda,\chi)$ is
carried, by a metaplectic change of phase-space coordinates, to an ideal
codeword for the rectangular lattice $\Lambda_D$.

If \(\chi_D^T\) does not coincide with the canonical rectangular
phase \(\chi_D\), their quotient
\[
\eta(\lambda)
:=
\chi_D(\lambda)\bigl(\chi_D^T(\lambda)\bigr)^{-1}
\]
is a character of \(\Lambda_D\). By nondegeneracy of the symplectic
pairing, there exists \(z_0\in\mathbb R^{2n}\), unique modulo
\(\Lambda_D^\circ\), such that
\[
\eta(\lambda)
=
e^{-2\pi i\sigma(\lambda,z_0)},
\qquad \lambda\in\Lambda_D.
\]
The Weyl displacement \(\rho(-z_0)\) then maps the
\(\chi_D^T\)-sector onto the canonical \(\chi_D\)-sector. Indeed, if
\(\psi\in\mathcal C_{\Lambda_D,\chi_D^T}\), then
\[
\rho(\lambda)\psi
=
\bigl(\chi_D^T(\lambda)\bigr)^{-1}\psi,
\]
and hence
\[
S_\lambda^{\chi_D}\rho(-z_0)\psi
=
\chi_D(\lambda)
e^{2\pi i\sigma(\lambda,z_0)}
\bigl(\chi_D^T(\lambda)\bigr)^{-1}
\rho(-z_0)\psi
=
\rho(-z_0)\psi.
\]
Consequently,
\[
\mathcal U_{\Lambda,\chi}
:=
\rho(-z_0)U_T^*
\]
defines an isomorphism
\[
\mathcal U_{\Lambda,\chi}:
\mathcal C_{\Lambda,\chi}
\longrightarrow
\mathcal C_{\Lambda_D,\chi_D}.
\]

We use this map to define the logical vector of
$\psi\in\mathcal C_{\Lambda,\chi}$: it is the unique vector
$c_\psi\in\mathcal H_{\mathbf d}$ such that
\[
    \mathcal Z_D
    \bigl(\mathcal  U_{\Lambda,\chi}\psi\bigr)
    =
    c_\psi\,\delta_0.
\]
By Proposition~\ref{prop:Zak-characterization}, this yields the
identifications $\mathcal C_{\Lambda,\chi}
    \cong
    \mathcal H_{\mathbf d}
    \cong
    \mathbb C^{d_\Lambda}$. The same transport also identifies the logical displacement operators.
Since $T$ is symplectic, then $\Lambda^\circ
    =
    T\Lambda_D^\circ$,
and therefore $T^{-1}$ induces an isomorphism
\[
    \Lambda^\circ/\Lambda
    \longrightarrow
    \Lambda_D^\circ/\Lambda_D,
    \qquad
    [\mu]\longmapsto[T^{-1}\mu].
\]
Let $\mu\in\Lambda^\circ$, and write
\[
    T^{-1}\mu
    =
    \nu_{q_\mu}+\theta_\mu,
    \qquad
    \theta_\mu\in\Lambda_D,
\]
where $\nu_{q_\mu}\in\Lambda_D^\circ$ is the chosen representative of
the class
\[
    q_\mu=[T^{-1}\mu]
    \in\Lambda_D^\circ/\Lambda_D.
\]
If $\mathcal U_{\Lambda,\chi}$ contains the additional Weyl correction
described above, its commutation with $\rho(\mu)$ contributes only a
scalar phase. Hence there exists $\beta_\mu\in\mathbb T$ such that
\[
    \mathcal U_{\Lambda,\chi}\rho(\mu)\psi
    =
    \beta_\mu\,
    \rho(T^{-1}\mu)
    \mathcal U_{\Lambda,\chi}\psi.
\]
Using \eqref{eq:weyl-commutation_0},
\[
    \rho(T^{-1}\mu)
    =
    e^{\pi i\sigma(\nu_{q_\mu},\theta_\mu)}
    \rho(\nu_{q_\mu})\rho(\theta_\mu).
\]
Since
$\mathcal U_{\Lambda,\chi}\psi\in
\mathcal C_{\Lambda_D,\chi_D}$,
we have
\[
    \rho(\theta_\mu)
    \mathcal U_{\Lambda,\chi}\psi
    =
    \overline{\chi_D(\theta_\mu)}
    \mathcal U_{\Lambda,\chi}\psi.
\]
Therefore
\begin{align*}
\mathcal Z_D
\bigl(
    \mathcal U_{\Lambda,\chi}\rho(\mu)\psi
\bigr)
&=
\beta_\mu
e^{\pi i\sigma(\nu_{q_\mu},\theta_\mu)}
\overline{\chi_D(\theta_\mu)}
\mathcal Z_D
\bigl(
    \rho(\nu_{q_\mu})
    \mathcal U_{\Lambda,\chi}\psi
\bigr)
\\
&=
\beta_\mu
e^{\pi i\sigma(\nu_{q_\mu},\theta_\mu)}
\overline{\chi_D(\theta_\mu)}
\bigl(
    \mathsf W_{q_\mu}c_\psi
\bigr)\delta_0.
\end{align*}
Thus, up to a phase, the action of the physical displacement
$\rho(\mu)$ on the logical information of $\psi$ is given by the finite
Weyl operator $\mathsf W_{q_\mu}$, where
\[
    q_\mu=[T^{-1}\mu]
    \in\Lambda_D^\circ/\Lambda_D.
\]

\section{Adjoint-lattice coefficients and logical reconstruction}
\label{sec:finite-block}
The preceding section described the ideal GKP code
\[
\mathcal C_{\Lambda,\chi}
\subset M^\infty(\mathbb R^n)
\]
through the vector-valued Zak transform and identified it with the
finite-dimensional logical space \(\mathcal H_{\mathbf d}\). We now
introduce a complementary description based on stable coefficient
expansions indexed by the adjoint lattice \(\Lambda^\circ\).

The two constructions should be distinguished. The ideal code consists
of distributional stabilizer eigenvectors and is finite-dimensional. By
contrast, the model spaces introduced below are infinite-dimensional
complemented subspaces of vector-valued modulation spaces. Their purpose
is to provide stable and unique adjoint-lattice coordinates in which the
stabilizer, logical, and displacement actions can be analyzed.

\subsection{The projected model spaces}
\label{subsec:projected-model-space}

Let $\Lambda$ be a lattice and let
$\mathbf g=(g_1,\ldots,g_m)
\in M^1(\mathbb R^n)\otimes\mathbb C^m
$ generate a multi-window Gabor frame
\[
\mathcal G(\mathbf g,\Lambda)
=
\{\rho(\lambda)g_r:
\lambda\in\Lambda,\ 1\leq r\leq m\}
\]
for \(L^2(\mathbb R^n)\). Let $\boldsymbol\gamma=(\gamma_1,\ldots,\gamma_m)
\in M^1(\mathbb R^n)\otimes\mathbb C^m$  be its canonical dual family. 

For \(1\leq p\leq\infty\), we  define the synthesis operator
\[
D_{\boldsymbol\gamma,\Lambda^\circ}a
:=
\sum_{\mu\in\Lambda^\circ}
a_\mu\widehat\rho(\mu)\boldsymbol\gamma,
\qquad
a=(a_\mu)_{\mu\in\Lambda^\circ}
\in\ell^p(\Lambda^\circ).
\]
For \(p<\infty\), the series converges in $M^p(\mathbb R^n)\otimes\mathbb C^m$.
For \(p=\infty\), it converges in the weak-* topology of $M^\infty(\mathbb R^n)\otimes\mathbb C^m$ .

\begin{definition}
For \(1\leq p\leq\infty\), the \emph{adjoint-lattice model space}
associated with \(\boldsymbol\gamma\in M^1(\mathbb R^n)\otimes\mathbb C^m\) is
\[
V_{\boldsymbol\gamma}^p
:=
D_{\boldsymbol\gamma,\Lambda^\circ}
(\ell^p(\Lambda^\circ))
\subset
M^p(\mathbb R^n)\otimes\mathbb C^m.
\]
\end{definition}

By (\ref{eq:model-space-stability})  every element of \(V_{\boldsymbol\gamma}^p\) has a unique and stable
coefficient expansion indexed by \(\Lambda^\circ\).

\subsection{Coefficient extraction and the model-space projection}
\label{subsec:model-projection}

We now identify the projection onto the adjoint-lattice model space with the
Heisenberg-module idempotent introduced in Section~\ref{sec:gabor}. Recall
that
\[
\mathsf P_{\boldsymbol \gamma, \mathbf g}
:=
\bigl({}_\Lambda\!\langle \gamma_r,g_s\rangle\bigr)_{r,s=1}^m
\in M_m(A_\Lambda)
\]
is an idempotent.

The analytic point needed here is that the fundamental identity of Gabor
analysis extends beyond the Hilbert-space setting. Since
$\mathbf g,\boldsymbol\gamma\in M^1(\mathbb R^n)\otimes\mathbb C^m$, the
associativity identity
\begin{equation}
\label{eq:module-associativity-Minf}
{}_\Lambda\!\langle \gamma_r,g_s\rangle\, f
=
\gamma_r\,\langle g_s,f\rangle_{\Lambda^\circ}
\end{equation}
continues to hold for $f\in M^\infty(\mathbb R^n)$, with the right-hand side
interpreted in the natural weak-* sense; see \cite{FeichtingerLuefFIGA}. More generally,
the corresponding analysis and synthesis maps are bounded on the usual
modulation-space scale.

For $1\leq p\leq\infty$, let $\mathbf f=(f_1,\ldots,f_m)
\in M^p(\mathbb R^n)\otimes\mathbb C^m$.  The matrix idempotent acts componentwise by
\begin{align*}
(\mathsf P_{\boldsymbol \gamma, \mathbf g}(\mathbf f))_r
&=
\sum_{s=1}^m
{}_\Lambda\!\langle\gamma_r,g_s\rangle f_s
\nonumber=
\sum_{s=1}^m
\gamma_r\langle g_s,f_s\rangle_{\Lambda^\circ}.
\end{align*}
Expanding the right $\Lambda^\circ$-valued inner product and reindexing the
adjoint lattice gives
\begin{equation}
\label{eq:module-idempotent-expansion}
(\mathsf P_{\boldsymbol \gamma, \mathbf g}(\mathbf f))_r
=
\sum_{\mu\in\Lambda^\circ}
\left(
\frac{1}{\operatorname{vol}(\Lambda)}
\sum_{s=1}^m
\langle f_s,\rho(\mu)g_s\rangle
\right)
\rho(\mu)\gamma_r.
\end{equation}

This formula motivates the adjoint-lattice coefficients
\begin{equation}
\label{eq:adjoint-lattice-coefficients}
a^{\mathbf g}_\mu(\mathbf f)
:=
\frac{1}{\operatorname{vol}(\Lambda)}
\sum_{s=1}^m
\langle f_s,\rho(\mu)g_s\rangle,
\qquad
\mu\in\Lambda^\circ.
\end{equation}
To simplify notation, we will denote it just by $a_\mu(\mathbf f)$ when the ${\mathbf g}$ is understood.
Accordingly, define
\begin{equation*}
\mathsf P_{\boldsymbol \gamma, \mathbf g}(\mathbf f)
:=
\sum_{\mu\in\Lambda^\circ}
a_\mu(\mathbf f)\,
\widehat\rho(\mu)\boldsymbol\gamma.
\end{equation*}
By \eqref{eq:module-idempotent-expansion}, this operator is exactly the action
of the module idempotent $\mathsf P_{\boldsymbol \gamma, \mathbf g}$ on
$M^p(\mathbb R^n)\otimes\mathbb C^m$. In particular, $\mathsf P_{\boldsymbol \gamma, \mathbf g}^2=\mathsf P_{\boldsymbol \gamma, \mathbf g}$,
and we identify its range
\[
\mathsf P_{\boldsymbol \gamma, \mathbf g}(M^p(\mathbb R^n)\otimes\mathbb C^m)
=
V_{\boldsymbol\gamma}^p.
\]
Thus, the stable adjoint-lattice coefficient model is precisely the range of
the Heisenberg-module idempotent.

For every $\mathbf f\in V_{\boldsymbol\gamma}^p$, the coefficient sequence in
\eqref{eq:adjoint-lattice-coefficients} is therefore the unique  $a(\mathbf f)\in \ell^{p}(\Lambda^\circ)$
satisfying
\[
\mathbf f
=
\sum_{\mu\in\Lambda^\circ}
a_\mu(\mathbf f)\,
\widehat\rho(\mu)\boldsymbol\gamma.
\]

\subsection{Adjoint-lattice covariance}
\label{subsec:adjoint-lattice-covariance}

The model space is invariant under the adjoint-lattice Weyl action. The
corresponding action on the coefficient sequence is a twisted
translation.

\begin{proposition}
\label{prop:coefficient-covariance}
Let $1\leq p\leq\infty$. Then for every \(\nu,\mu\in\Lambda^\circ\) 
\begin{equation}
\label{eq:coefficient-covariance}
a_\mu\bigl(\widehat\rho(\nu)\mathbf f\bigr)
=
e^{-\pi i\sigma(\nu,\mu)}
a_{\mu-\nu}(\mathbf f), \qquad \mathbf f 
\in M^p(\mathbb R^n)\otimes\mathbb C^m.
\end{equation}
\end{proposition}

\begin{proof}
We have 
\begin{align*}
a_\mu\bigl(\widehat\rho(\nu)\mathbf f\bigr)
&=
\frac{1}{\operatorname{vol}(\Lambda)}
\sum_{r=1}^m
\langle
\rho(\nu)f_r,\rho(\mu)g_r
\rangle
=
\frac{1}{\operatorname{vol}(\Lambda)}
\sum_{r=1}^m
\langle
f_r,\rho(-\nu)\rho(\mu)g_r
\rangle.
\end{align*}
Then by \eqref{eq:weyl-commutation_0} it follows the desired equality $
a_\mu\bigl(\widehat\rho(\nu)\mathbf f\bigr)
=
e^{-\pi i\sigma(\nu,\mu)}
a_{\mu-\nu}(\mathbf f)$.

\end{proof}

\begin{proposition}
\label{prop:model-projection-covariance}
Let $1\leq p\leq\infty$. For every \(\nu\in\Lambda^\circ\), $$\mathsf P_{\boldsymbol \gamma, \mathbf g}\widehat\rho(\nu)
=
\widehat\rho(\nu)\mathsf P_{\boldsymbol \gamma, \mathbf g}\,.$$
Consequently, $\widehat\rho(\nu)V_{\boldsymbol\gamma}^p
=
V_{\boldsymbol\gamma}^p$.
\end{proposition}

\begin{proof}
Using \eqref{eq:coefficient-covariance} and reindexing
\(\mu=\kappa+\nu\), we obtain
\begin{align*}
\mathsf P_{\boldsymbol \gamma, \mathbf g}(\widehat\rho(\nu)\mathbf f)
&=
\sum_{\mu\in\Lambda^\circ}
e^{-\pi i\sigma(\nu,\mu)}
a_{\mu-\nu}(\mathbf f)
\widehat\rho(\mu)\boldsymbol\gamma
\\
&=
\sum_{\kappa\in\Lambda^\circ}
e^{-\pi i\sigma(\nu,\kappa)}
a_\kappa(\mathbf f)
\widehat\rho(\nu+\kappa)\boldsymbol\gamma
\\
&=
\widehat\rho(\nu)
\sum_{\kappa\in\Lambda^\circ}
a_\kappa(\mathbf f)
\widehat\rho(\kappa)\boldsymbol\gamma
\\
&=
\widehat\rho(\nu)\mathsf P_{\boldsymbol \gamma, \mathbf g}(\mathbf f).
\end{align*}
In the second equality, we used
\(\sigma(\nu,\nu)=0\), and in the third we used \eqref{eq:weyl-commutation_0}.
\end{proof}

\subsection{Realization of ideal codewords in the model space}

For \(\alpha=(\alpha_1,\ldots,\alpha_m)\in\C^m\), define
\[
J_\alpha (f):=(\alpha_1f,\ldots,\alpha_mf)=f\otimes \alpha\in M^\infty(\mathbb R^n)\otimes\mathbb C^m,\qquad f\in M^\infty(\mathbb R^n) 
\]
and
\begin{equation*}
\Phi_{\gamma,\alpha}
:=\mathsf P_{\boldsymbol \gamma, \mathbf g}\circ J_\alpha:
M^\infty(\R^n)\longrightarrow V_\gamma^\infty.
\end{equation*}
Then  we have that 
\begin{align*}
    \Phi_{\boldsymbol \gamma,\alpha}(f) & = \mathsf P_{\boldsymbol \gamma, \mathbf g}(\alpha_1f,\ldots,\alpha_mf) \\
    & = \bigl({}_{\Lambda}\!\langle\gamma_r,g_s\rangle\bigr)_{r,s=1}^m(\alpha_1f,\ldots,\alpha_mf)^T \\
    & = \bigl(\sum_{s=1}^m {}_{\Lambda}\langle \gamma_r,g_s\rangle \alpha_s f\bigr)_{r=1}^m \\
    & = \bigl(\sum_{s=1}^m   \gamma_r\langle g_s, \alpha_s f\rangle_{\Lambda^\circ}\bigr)_{r=1}^m \\
     & = \bigl(   \gamma_r\langle \sum_{s=1}^m \overline{\alpha_s}g_s,  f\rangle_{\Lambda^\circ}\bigr)_{r=1}^m \,.
\end{align*}
In the last equality we used (\ref{eq:module-associativity-Minf}).
Now  put
\begin{equation*}
h_\alpha:=\sum_{s=1}^m\overline{\alpha_s}g_s,
\end{equation*}
then
\begin{align*}
    \Phi_{\boldsymbol \gamma,\alpha}(f) & = \bigl(    \gamma_r\langle h_\alpha,  f\rangle_{\Lambda^\circ}\bigr)_{r=1}^m \\
    & = \bigl( \frac1{\vol(\Lambda)}\sum_{\mu\in \Lambda^\circ} \langle f,\rho(\mu)^* h_\alpha\rangle\rho(\mu)^* \gamma_r \bigr)_{r=1}^m \\
    & = \sum_{\mu\in \Lambda^\circ} \frac1{\vol(\Lambda)} \langle f,\rho(-\mu) h_\alpha\rangle \widehat\rho(-\mu)\boldsymbol \gamma
\end{align*}
If we define 
\begin{equation}
\label{eq:coeff}
a_{\mu,\alpha}(f)
:=\frac{1}{\vol(\Lambda)}
\ip{f}{\rho(\mu)h_\alpha},
\end{equation}
then we have that 
\begin{equation}
\label{eq:scalar-probe-coeff}
 \Phi_{\boldsymbol \gamma,\alpha}(f)=\sum_{\mu\in \Lambda^\circ} a_{\mu,\alpha}(f) \widehat\rho(\mu)\boldsymbol \gamma. 
 \end{equation}
  where $\{a_{\mu,\alpha}:\mu\in \Lambda^\circ\}$ are the unique coefficients satisfying (\ref{eq:scalar-probe-coeff}) (see Section \ref{subsec:model-projection}).    

Now, if \(f\in\cC_{\Lambda,\chi}\), then we have that $\chi(\lambda)\rho(\lambda)f=f$ for every $\lambda\in\Lambda$, so then using linearity of $\Phi_{\boldsymbol \gamma,\alpha}$ and Proposition \ref{prop:model-projection-covariance} we have that 
\begin{align*}
    \Phi_{\boldsymbol \gamma,\alpha} ( f) & = \Phi_{\boldsymbol \gamma,\alpha} (\chi(\lambda)\rho(\lambda) f) \\
    & =\chi(\lambda)\widehat\rho(\lambda) \Phi_{\boldsymbol \gamma,\alpha} ( f) \\
    & = \chi(\lambda)\widehat\rho(\lambda) \sum_{\mu\in \Lambda^\circ} a_{\mu,\alpha}(f) \widehat\rho(\mu)\boldsymbol \gamma \\ 
    & =  \sum_{\mu\in \Lambda^\circ} e^{-\pi i\sigma(\lambda,\mu)}\chi(\lambda)a_{\mu,\alpha}(f) \widehat\rho(\mu+\lambda)\boldsymbol \gamma \\ 
    & =  \sum_{\mu\in \Lambda^\circ} e^{-\pi i\sigma(\lambda,\mu-\lambda)}\chi(\lambda)a_{\mu-\lambda,\alpha}(f) \widehat\rho(\mu)\boldsymbol \gamma \\
     & =  \sum_{\mu\in \Lambda^\circ} e^{-\pi i\sigma(\lambda,\mu)}{\chi(\lambda)}a_{\mu-\lambda,\alpha}(f) \widehat\rho(\mu)\boldsymbol \gamma \,.
\end{align*}
But by the uniqueness of coefficients we have that 
\begin{equation}\label{eq:stabilizer-coefficient-covariance}
    e^{-\pi i\sigma(\lambda,\mu)}{\chi(\lambda)}a_{\mu-\lambda,\alpha}(f) = a_{\mu,\alpha}(f)\qquad
\lambda\in\Lambda,\ \mu\in\Lambda^\circ.
\end{equation}

Thus, one representative of every coset in \(\Lambda^\circ/\Lambda\)
determines the entire coefficient sequence.

\subsection{Finite adjoint-lattice blocks recover the logical vector}

We now establish the main link between the Zak and Gabor descriptions.
We give the statement in rectangular coordinates; the general integral
lattice case follows by metaplectic transport.

Let $\Lambda=\Lambda_D$. The projective representation \(K_D\ni q\mapsto\mathsf W_q\) (see (\ref{pauli_op})) is irreducible.
Indeed, 
\[
\Span\{\mathsf W_q:q\in K_D\}
=
\Span\{T:T\in \mathcal P_{\mathbf d}\}
=
B(\mathcal H_{\mathbf d}).
\]
Consequently,
\[
\{\mathsf W_q:q\in K_D\}'
=
B(\mathcal H_{\mathbf d})'
=
\mathbb C I.
\]
It follows that the projective representation is irreducible.

Let  $G(x,\omega)=
\bigl(\mathcal Z_Dg_1(x,\omega),\ldots,
      \mathcal Z_Dg_m(x,\omega)\bigr)$, be the $d\times m$ Zak matrix. For the probe $\alpha\in\C^m$, define
\begin{equation*}
v_\alpha:=G(0,0)\overline\alpha
=\mathcal Z_Dh_\alpha(0,0)
\in\mathcal H_{\mathbf d}=\bigotimes_{i=1}^n \mathbb C^{d_i}.
\end{equation*}
Recall that  \(K_D\) may be naturally identified with the finite phase space $\mathbb Z_{\mathbf d}\times\mathbb Z_{\mathbf d}$ (see \eqref{identification_K}).

\begin{lemma}
\label{lem:finite-weyl-tight-frame}
Let \(0\neq v\in \mathcal H_{\mathbf d}\). Then
\begin{equation}
\label{eq:finite-weyl-reconstruction}
c=
\frac{1}{d\|v\|^2}
\sum_{q\in K_D}
\langle c,W_qv\rangle W_qv
\end{equation}
for every $c\in \mathcal H_{\mathbf d}$.
Consequently,
\begin{equation*}
\sum_{q\in K_D}
|\langle c,W_qv\rangle|^2
=
d\|v\|^2\|c\|^2.
\end{equation*}
\end{lemma}

\begin{proof}
Define the frame operator
\[
S_v c
=
\sum_{q\in K_D}
\langle c,W_qv\rangle W_qv,
\qquad c\in\mathcal H_{\mathbf d}.
\]
Equivalently,
\[
S_v
=
\sum_{q\in K_D}
W_q |v\rangle\langle v| W_q^*.
\]

We claim that \(S_v\) commutes with every Pauli operator. Indeed, for
\(p\in K_D\),
\[
W_p S_v W_p^*
=
\sum_{q\in K_D}
W_pW_q |v\rangle\langle v| W_q^*W_p^*\sum_{q\in K_D}
W_{p+q}|v\rangle\langle v|W_{p+q}^*
=
S_v.
\]

Since the finite Weyl representation acts irreducibly on \(\mathcal H_{\mathbf d}\), we have that $S_v$ must be a scalar multiple of the identity
\(
S_v=\alpha I
\)
for some scalar \(\alpha\geq 0\).

To determine \(\alpha\), we take traces. Since each \(W_q\) is unitary,
\[
\operatorname{Tr}(S_v)
=
\sum_{q\in K_D}
\operatorname{Tr}
\bigl(W_q |v\rangle\langle v| W_q^*\bigr)
=
\sum_{q\in K_D}\|v\|^2.
\]
Using \(|K_D|=d^2\), we obtain $\operatorname{Tr}(S_v)=d^2\|v\|^2$.
On the other hand, since
\(\dim\mathcal H_{\mathbf d}=d\),
\[
\operatorname{Tr}(S_v)=\alpha d.
\]
Therefore $\alpha=d\|v\|^2$,  and hence
\[
S_v=d\|v\|^2I.
\]
This proves \eqref{eq:finite-weyl-reconstruction}.
\end{proof}

\begin{theorem}[Finite-block logical reconstruction]
\label{thm:finite-block-recovery}
Let \(\mathbf g\in M^1(\R^n)\otimes\C^m\) generate a multi-window Gabor frame over
\(\Lambda_D\), let \(\boldsymbol\gamma\) be its canonical dual, and fix
\(\alpha\in\C^m\). For \(\psi_c\in\cC_{\Lambda_D}\) with
\(\cZ_D\psi_c=c\delta_0\),  we have that
\begin{equation*}
a_{\nu_q}(\Phi_{\boldsymbol\gamma,\alpha}(\psi_c))
=\frac1{ d}\ip{c}{\mathsf W_qv_\alpha},
\qquad q\in K_D.
\end{equation*}
Consequently,
\begin{equation*}
\sum_{q\in K_D}\abs{B_0(c)_q}^2
= \norm{c}^2.
\end{equation*}
The following are equivalent:
\begin{enumerate}[label=(\roman*)]
\item \(v_\alpha\neq0\);
\item the finite-block analysis map
\[
B_0:\mathcal H_{\mathbf d}\longrightarrow\ell^2(K_D),
\qquad
(B_0(c))_q
:=
\frac1{\sqrt d\|v_\alpha\|} \langle c,\mathsf W_qv_\alpha\rangle,
\]
is an isometry.
\item \(\Phi_{\gamma,\alpha}\) is injective on \(\cC_{\Lambda_D}\).
\end{enumerate}
When these conditions hold,
\begin{equation*}
c=\frac{1}{\sqrt d\norm{v_\alpha}}
\sum_{q\in K_D}B_0(c)_q\mathsf W_qv_\alpha.
\end{equation*}
\end{theorem}

\begin{proof}
By 
\eqref{eq:coeff} and Zak duality \eqref{eq:vector-zak-duality}, we have that 
\[
\begin{aligned}
a_{\nu_q}(\Phi_{\boldsymbol\gamma,\alpha}(\psi_c))
&=\frac{1}{d}\ip{\psi_c}{\rho(\nu_q)h_\alpha}\\
&=\frac{1}{d}\ip{c}{\cZ_D(\rho(\nu_q)h_\alpha)(0,0)}\\
&=\frac{1}{d}\ip{c}{\mathsf W_q\cZ_D(h_\alpha)(0,0)}
=\frac{1}{d}\ip{c}{\mathsf W_q v_\alpha}.
\end{aligned}
\]
A straightforward computation shows that 
\[
B_0^* (b)
=
\frac1{\sqrt d \|v_\alpha\|}
\sum_{p\in K_D}b_p\mathsf W_pv_\alpha \qquad b\in \ell^2(K_D).
\]
Then using  
Lemma~\ref{lem:finite-weyl-tight-frame} it follows that  $B_0^*B_0
=
I_{\mathcal H_{\mathbf d}}$,
thus, $B_0$ is an isometry.

\end{proof}

The normalized block map identifies the logical Hilbert space
with the distinguished $d$-dimensional subspace
\[
\mathcal B_0:=\operatorname{Ran}(B_0)
\subset\ell^2(K_D).
\]
Consequently, an arbitrary coefficient block
$b\in\ell^2(K_D)$ can be decomposed uniquely into an admissible
logical component and an orthogonal inconsistency component.
This gives a canonical coefficient-space reconstruction map; it
is not, by itself, a quantum recovery channel.

\begin{proposition}[Projection and logical recovery]
\label{prop:block-projection}
Assume the hypotheses of Theorem~\ref{thm:finite-block-recovery}
and let
\[
    \mathcal B_0:=\operatorname{Ran}(B_0)
    \subseteq \ell^2(K_D)
\]
be the space of admissible exact coefficient blocks. Then $\Pi_0:=B_0B_0^*$ is the orthogonal projection of $\ell^2(K_D)$ onto $\mathcal B_0$. Consequently, for every $b\in\ell^2(K_D)$, $\widehat b:=\Pi_0b$ is the unique solution of
\[
    \widehat b
    =
    \operatorname*{argmin}_{a\in\mathcal B_0}
    \|b-a\|_{\ell^2(K_D)},
\]
and the corresponding logical vector is $\widehat c:=B_0^*b$. In particular, $\widehat b=B_0\widehat c$.
\end{proposition}

\begin{proof}
Since $B_0$ is an isometry by
Theorem~\ref{thm:finite-block-recovery}, we have
\[
    B_0^*B_0=I_{\mathcal H_d}.
\]
It follows that
\[
    \Pi_0^2
    =B_0B_0^*B_0B_0^*
    =B_0B_0^*
    =\Pi_0
\]
and $\Pi_0^*=\Pi_0$. Its range is precisely
$\operatorname{Ran}(B_0)$, so it is the orthogonal projection onto
$\mathcal B_0$. The minimization property follows from the
Hilbert-space projection theorem, and
\[
    B_0^*\Pi_0b
    =B_0^*B_0B_0^*b
    =B_0^*b.
\]
\end{proof}

\begin{corollary}[Periodicity of the finite coefficient blocks]
\label{cor:periodicity-finite-blocks}
Assume the hypotheses of
Theorem~\ref{thm:finite-block-recovery}. For
$\lambda\in\Lambda_D$, define
\[
B_\lambda(c)
:=
\frac{\sqrt d}{\|v_\alpha\|}\left(
e^{\pi i\sigma(\lambda,\nu_q)}
a_{\lambda+\nu_q}
\bigl(\Phi_{\gamma,\alpha}(\psi_c)\bigr)
\right)_{q\in K_D}.
\]
Then
\[
B_\lambda(c)
=
\chi_D(\lambda)B_0(c).
\]
In particular, every phase-corrected block contains exactly the same
logical information, up to the known scalar phase
$\chi_D(\lambda)$.
\end{corollary}

\begin{proof}
By Proposition~\ref{prop:coefficient-covariance} and the
stabilizer relation,
\[
a_{\lambda+\nu_q}
\bigl(\Phi_{\gamma,\alpha}(\psi_c)\bigr)
=
\chi_D(\lambda)
e^{-\pi i\sigma(\lambda,\nu_q)}
a_{\nu_q}
\bigl(\Phi_{\gamma,\alpha}(\psi_c)\bigr).
\]
Multiplication by
$e^{\pi i\sigma(\lambda,\nu_q)}$ gives $\bigl(B_\lambda(c)\bigr)_q
=
\chi_D(\lambda)\bigl(B_0(c)\bigr)_q$.
\end{proof}

\begin{corollary}[Probe design]
\label{cor:probe-design}
For every multi-window Gabor frame as above, there exists
\(\alpha\in\C^m\) for which \(\Phi_{\gamma,\alpha}\) is injective on the ideal
code. More precisely, every prescribed nonzero
\(v\in\cH_{\mathbf d}\) can be realized as \(v=G(0,0)\overline{\alpha}\).
\end{corollary}
\begin{proof}
When $\mathcal G(\mathbf g,\Lambda_D)$ is a Gabor frame then  the map $G(0,0):\mathcal C^m\to  \mathcal H_{\mathbf d}$ is surjective (see Section \ref{Zak_matrix}).
\end{proof}

\begin{proposition}[Clifford action on   coefficient blocks]
\label{prop:clifford-action-fixed-probe-blocks}
Assume the hypotheses of
Theorem~\ref{thm:finite-block-recovery}, with
\(v_\alpha\neq0\), and let $B_0:\mathcal H_{\mathbf d}\longrightarrow\ell^2(K_D)$ be the finite-block analysis map. Suppose that \(T\in\operatorname{Sp}(2n,\mathbb R)\) satisfies $T\Lambda_D=\Lambda_D$ and $\chi_D\circ T^{-1}=\chi_D$,
and let \(H_T\) be the logical Clifford unitary induced by a
metaplectic lift \(U_T\). Then \(H_T\) induces an operator
\[
\mathcal C_{T}
:=
B_0 H_TB_0^*:\ell^2(K_D)\to \ell^2(K_D),
\]
that restricts to a unitary operator  $\mathcal C_T:\Ran(B_0)\to \Ran(B_0)$.
The matrix coefficients of \(\mathcal C_{T}\) are
\begin{equation*}
\label{eq:fixed-probe-clifford-matrix}
\mathcal C_{T}(q,p)
=
\frac{1}{d\|v_\alpha\|^2}
\left\langle
H_T\mathsf W_pv_\alpha,
\mathsf W_qv_\alpha
\right\rangle,
\qquad p,q\in K_D.
\end{equation*}
\end{proposition}

\begin{proof}
By Theorem~\ref{thm:finite-block-recovery}, the reconstruction formula gives
\[
c
=
\frac1{\sqrt d\|v_\alpha\|}
\sum_{p\in K_D} B_0(c)_p
\mathsf W_pv_\alpha.
\]
Therefore,
\begin{align*}
B_0(H_Tc)_q
&=
\frac1{\sqrt d \|v_\alpha\|}
\left\langle H_Tc,\mathsf W_qv_\alpha\right\rangle \\
&=
\frac{1}{d\|v_\alpha\|^2}
\sum_{p\in K_D}
B_0(c)_p
\left\langle
H_T\mathsf W_pv_\alpha,
\mathsf W_qv_\alpha
\right\rangle,
\end{align*}
which proves the desired result.
\end{proof}

\begin{remark}[Clifford action on coefficient blocks]
Retain the notation of Proposition~\ref{prop:clifford-action-fixed-probe-blocks}, and let
\[
\tau_T:K_D\longrightarrow K_D, \qquad \nu_q\longmapsto
T\nu_q\equiv\nu_{\tau_T(q)}\pmod{\Lambda_D}.
\]
By the Clifford covariance established in
Section~\ref{Clifford_gates},
\[
H_T\mathsf W_qH_T^*=e^{i\theta_q}\mathsf W_{\tau_T(q)}.
\]

Since $U_T^*\psi_c=\psi_{H_T^*c}$, we obtain
\begin{align*}
a_{\nu_q}\bigl(\Phi_{\gamma,\alpha}(U_T^*\psi_c)\bigr)
&=
\frac1d
\left\langle
U_T^*\psi_c,\rho(\nu_q)h_\alpha
\right\rangle
\\
&=
\frac1d
\left\langle
H_T^*c,
\mathcal Z_D\bigl(\rho(\nu_q)h_\alpha\bigr)(0,0)
\right\rangle
\\
&=
\frac1d
\left\langle
H_T^*c,\mathsf W_qv_\alpha
\right\rangle
\\
&=
\frac1d
\left\langle
c,H_T\mathsf W_qv_\alpha
\right\rangle
\\
&=
e^{-i\theta_q}\frac1d
\left\langle
c,\mathsf W_{\tau_T(q)}H_Tv_\alpha
\right\rangle,
\end{align*}
where $v_\alpha=\mathcal Z_Dh_\alpha(0,0)$. In particular, suppose that \(v_\alpha\) is an eigenvector of \(H_T\),
say
\[
H_Tv_\alpha=e^{i\varphi}v_\alpha,
\]
for some $\varphi\in \mathbb R$. Then
\[
a_{\nu_q}\bigl(\Phi_{\gamma,\alpha}(U_T^*\psi_c)\bigr)
=
e^{-i(\theta_q+\varphi)}
a_{\nu_{\tau_T(q)}}\bigl(\Phi_{\gamma,\alpha}(\psi_c)\bigr).
\]
Thus,  the adjoint logical Clifford gate acts on the
coefficient block by the permutation of \(K_D\) induced by \(S\),
together with coefficient-dependent phases.  
\end{remark}

In the following corollary we are going to show that it is not necessary to sample the signal over all coefficients in $K_D$ to recover the logical information, and that a minimal number of samples is enough if we choose the probe appropriately.

Recall that $\mathcal H_{\mathbf d}
=
\bigotimes_{j=1}^n\mathbb C^{d_j}$ and that $K_D$ can be decomposed as $\prod_{j=1}^n\mathbb Z_{d_j}^2$.

\begin{corollary}[Explicit minimal sampling for
Gaussian probes]
Let 
\[
v_\alpha
=
\mathsf h_{d_1}^{\lambda_1}
\otimes\cdots\otimes
\mathsf h_{d_n}^{\lambda_n},
\qquad
\lambda_j>0,
\]
where \(\mathsf h_{d_j}^{\lambda_j}\in\mathbb C^{d_j}\)
is the sampled and periodized Gaussian of
\cite{AbreuBalazsHolighausLuefSpeckbacher}.

For each \(j=1,\ldots,n\), let
\[
Q_j\subseteq\mathbb Z_{d_j}^2,
\qquad
|Q_j|=d_j,
\]
consist of distinct points and assume that either \(d_j\) is odd or
\[
\sum_{q_j\in Q_j}q_j\neq(0,0)
\qquad\text{in }\mathbb Z_{d_j}^2.
\]
Define $Q:=Q_1\times\cdots\times Q_n\subseteq K_D$. Then $\{\mathsf W_qv_\alpha:q\in Q\}$ is a basis of \(\mathcal H_{\mathbf d}\). Consequently, the restricted
coefficient family
\[
\{B_0(c)_q:q\in Q\}
\]
determines every \(c\in\mathcal H_{\mathbf d}\) using the minimal
possible number \(d\) of scalar coefficients.
\end{corollary}

\begin{proof}
For each \(j\), Theorem~3 of
\cite{AbreuBalazsHolighausLuefSpeckbacher} shows that
\[
\bigl\{
\mathsf W_{q_j}^{(j)}
\mathsf h_{d_j}^{\lambda_j}
:q_j\in Q_j
\bigr\}
\]
is a basis of \(\mathbb C^{d_j}\). Under the identification $\mathcal H_{\mathbf d}
=
\bigotimes_{j=1}^n\mathbb C^{d_j}$,
the multimode Weyl operators factor  as
\[
\mathsf W_q
=
\mathsf W_{q_1}^{(1)}
\otimes\cdots\otimes
\mathsf W_{q_n}^{(n)},
\qquad
q=(q_1,\ldots,q_n)\in K_D.
\]
Therefore, $\{\mathsf W_qv_\alpha:q\in Q\}$  is the tensor product of the bases
associated with the sets \(Q_j\). It is consequently a basis of
\(\mathcal H_{\mathbf d}\).

Since $|Q|= \dim\mathcal H_{\mathbf d}$, this number of scalar coefficients is minimal. The reconstruction
claim follows from
\[
B_0(c)_q
=
\frac1{\sqrt d\|v_\alpha\|}
\langle c,\mathsf W_qv_\alpha\rangle.
\]
\end{proof}

\subsection{General symplectically integral lattices}

Let \((\Lambda,\chi)\) be symplectically integral, and choose a
symplectic normal form
\[
\Lambda=T\Lambda_D,
\]
and hence  $\Lambda^\circ=T\Lambda_D^\circ$. Let $\mathcal U_{\Lambda,\chi}$ be the unitary defined in Subsection \ref{subsec:general-zak} such that induces an isomorphism $\mathcal U_{\Lambda,\chi}:
\mathcal C_{\Lambda,\chi}
\longrightarrow
\mathcal C_{\Lambda_D,\chi_D}$.
 This unitary can be written as $\mathcal U_{\Lambda,\chi}=\rho(-z_0)U^*_T$, where $U_T$ is the metaplectic unitary associated to $T$ and for some $z_0\in \mathbb R^{2n}$.

Transport the Gabor windows, dual windows, and probes by
\[
g_i^\Lambda=\mathcal U^*_{\Lambda,\chi}g_i,\qquad
\gamma_i^\Lambda=\mathcal U^*_{\Lambda,\chi}\gamma_i,\qquad
h_\alpha^\Lambda=\mathcal U^*_{\Lambda,\chi}h_\alpha.
\]
Then \(\cG(\mathbf g^\Lambda,\Lambda)\) is a multi-window Gabor frame
with dual windows \(\boldsymbol\gamma^\Lambda=(\gamma_1^\Lambda,\ldots,\gamma_m^\Lambda)\).

The map \(q\mapsto T\nu_q+\Lambda\) identifies \(K_D\) with
\(\Lambda^\circ/\Lambda\).

\begin{corollary}[General symplectically integral lattices]
\label{cor:general-lattice-block}
For \(c\in\mathcal H_{\mathbf d}\), set $\psi_c^\Lambda:=\mathcal U_{\Lambda,\chi}\psi_c^D$.
Then, for every \(q\in K_D\),
\[
a_{T\nu_q}
 \bigl(
   \Phi_{\boldsymbol\gamma^\Lambda,\alpha}
        (\psi_c^\Lambda)
 \bigr)
=
e^{-2\pi i\sigma(z_0,\nu_q)}
a_{\nu_q}
 \bigl(
   \Phi_{\boldsymbol\gamma,\alpha}(\psi_c^D)
 \bigr).
\]

Assume that \(v_\alpha\neq0\), and define the phase-corrected block
\[
\mathcal B_{\Lambda,\alpha}(\psi_c^\Lambda)_q
:=
e^{2\pi i\sigma(z_0,\nu_q)}
\frac{\sqrt{\operatorname{vol}(\Lambda)}}{\|v_\alpha\|}
a_{T\nu_q}
 \bigl(
   \Phi_{\boldsymbol\gamma^\Lambda,\alpha}
        (\psi_c^\Lambda)
 \bigr),
\qquad q\in K_D.
\]
Then $\mathcal B_{\Lambda,\alpha}(\psi_c^\Lambda)=B_0(c)$,
and hence
\[
\bigl\|
\mathcal B_{\Lambda,\alpha}(\psi_c^\Lambda)
\bigr\|_{\ell^2(K_D)}
=
\|c\|_{\mathcal H_{\mathbf d}}.
\]
\end{corollary}
\begin{proof}
Since \(T\) is symplectic, $\operatorname{vol}(\Lambda)
=
\operatorname{vol}(\Lambda_D)$. Moreover, metaplectic covariance gives
$U_T^*\rho(T\nu_q)U_T=\rho(\nu_q)$. Since \(\mathcal U^*_{\Lambda,\chi}=U_T\rho(z_0)\), we obtain
\[
\begin{aligned}
U_{\Lambda,\chi}\rho(T\nu_q)U^*_{\Lambda,\chi}
&=
\rho(z_0)^*
U_T^*\rho(T\nu_q)U_T
\rho(z_0)\\
&=
\rho(z_0)^*\rho(\nu_q)\rho(z_0)\\
&=
e^{-2\pi i\sigma(z_0,\nu_q)}\rho(\nu_q).
\end{aligned}
\]
Therefore,
\[
\begin{aligned}
&
a_{T\nu_q}
 \bigl(
   \Phi_{\boldsymbol\gamma^\Lambda,\alpha}
        (\psi_c^\Lambda)
 \bigr)
\\
&\quad=
\frac{1}{\operatorname{vol}(\Lambda)}
\left\langle
\mathcal U^*_{\Lambda,\chi}\psi_c^D,
\rho(T\nu_q)V_Th_\alpha
\right\rangle
\\
&\quad=
\frac{1}{\operatorname{vol}(\Lambda_D)}
\left\langle
\psi_c^D,
\mathcal U_{\Lambda,\chi}\rho(T\nu_q)V_Th_\alpha
\right\rangle
\\
&\quad=
e^{-2\pi i\sigma(z_0,\nu_q)}
\frac{1}{\operatorname{vol}(\Lambda_D)}
\left\langle
\psi_c^D,\rho(\nu_q)h_\alpha
\right\rangle
\\
&\quad=
e^{-2\pi i\sigma(z_0,\nu_q)}
a_{\nu_q}
 \bigl(
   \Phi_{\boldsymbol\gamma,\alpha}(\psi_c^D)
 \bigr).
\end{aligned}
\]
Multiplying componentwise by
\(e^{2\pi i\sigma(z_0,\nu_q)}\), and using
\(\operatorname{vol}(\Lambda)=\operatorname{vol}(\Lambda_D)\), gives $\mathcal B_{\Lambda,\alpha}(\psi_c^\Lambda)=B_0(c)$.

The norm identity now follows from
Theorem~\ref{thm:finite-block-recovery}.
\end{proof}

\subsection{Example: the single-mode ideal comb}

Let
\[
\Lambda_d
=
\sqrt d\,\Z\times\sqrt d\,\Z,
\qquad
\Lambda_d^\circ
=
\frac1{\sqrt d}\Z\times\frac1{\sqrt d}\Z,
\qquad d\ge2,
\]
and consider the ideal comb
\[
f_0(x)
=
\sum_{\ell\in\Z}
\delta_0(x-\ell\sqrt d).
\]
Then \(f_0\in\cC_{\Lambda_d}\), and its Zak transform is concentrated
at the origin of the syndrome torus with logical vector \(e_0\in\C^d\), i.e., $\cZ_df_0=e_0\,\delta_0$.

Let $\mathbf g=(g_1,\ldots,g_m)\in M^1(\R)\otimes\C^m$
generate a multi-window Gabor frame over \(\Lambda_d\), and let
\[
G(x,\omega)
=
\bigl(
\cZ_dg_1(x,\omega),\ldots,\cZ_dg_m(x,\omega)
\bigr)
\]
be its \(d\times m\) Zak matrix. For the unweighted probe
\(\mathbf 1_m=(1,\ldots,1)\in\C^m\), set
\[
v
:=
G(0,0)\mathbf1_m
=
\sum_{j=1}^m\cZ_dg_j(0,0)
\in\C^d.
\]

In the present single-mode setting, the coordinates of this logical
probe can be written directly in terms of the Gabor windows. Indeed,
by the definition of the vector-valued Zak transform,
\[
\bigl(\cZ_dg_j(0,0)\bigr)_k
=
Z_dg_j\left(\frac{k}{\sqrt d},0\right),
\qquad
k\in\Z_d.
\]
Since $Z_dg_j(x,0)
=
\sum_{\ell\in\Z}
g_j(x-\ell\sqrt d)$,
we obtain
\begin{equation*}
G_{k,j}(0,0)
=
\sum_{\ell\in\Z}
g_j\left(
\frac{k-\ell d}{\sqrt d}
\right),
\qquad
k\in\Z_d,\quad j=1,\ldots,m.
\end{equation*}
Consequently,
\begin{equation}
\label{eq:single-mode-logical-probe}
v_k
=
\sum_{j=1}^m
\sum_{\ell\in\Z}
g_j\left(
\frac{k-\ell d}{\sqrt d}
\right),
\qquad
k\in\Z_d.
\end{equation}
Thus \(v\) consists of the \(d\) periodized samples of the sum of the
Gabor windows at the representatives
\[
0,\frac1{\sqrt d},\ldots,\frac{d-1}{\sqrt d}.
\]

We now compute the distinguished coefficient block of \(f_0\).
For $q=(r,s)\in
K_d
\cong
\Z_d\times\Z_d$,
the finite-block formula gives
\begin{equation*}
B_0(e_0)_{r,s} 
=
\frac1{\sqrt d \|v\|}
\ip{e_0}{\mathsf W_{r,s}v},
\end{equation*}
where
\[
\mathsf W_{r,s}
=
e^{\pi i rs/d}X^rZ^s.
\]
Since
\[
Xe_k=e_{k+1},
\qquad
Ze_k=e^{2\pi i k/d}e_k,
\]
we have
\[
\bigl(\mathsf W_{r,s}v\bigr)_0
=
e^{-\pi i rs/d}v_{-r},
\]
where indices are understood modulo \(d\). Then $\ip{e_0}{\mathsf W_{r,s}v}
=
e^{\pi i rs/d}\overline{v_{-r}}$. Therefore
\begin{equation*}
B_0(e_0)_{r,s} 
=
\frac1{\sqrt d \|v\|}
e^{\pi i rs/d}
\overline{v_{-r}}.
\end{equation*}
Using \eqref{eq:single-mode-logical-probe}, this becomes the completely
explicit formula
\begin{equation*}
B_0(e_0)_{r,s} 
=
\frac1{\sqrt d \|v\|}
e^{\pi i rs/d}
\sum_{j=1}^m
\sum_{\ell\in\Z}
\overline{
g_j\left(
\frac{\ell d-r}{\sqrt d}
\right)
}.
\end{equation*}

Hence the dependence of the coefficient block separates very clearly.
The codeword \(f_0\) contributes only its logical value \(e_0\), while
the choice of Gabor frame determines the \(d\) numbers
\[
A_r(g)
:=
\sum_{j=1}^m
\sum_{\ell\in\Z}
g_j\left(
\frac{\ell d-r}{\sqrt d}
\right),
\qquad
r\in\Z_d.
\]
In terms of these periodized samples,
\[
B_0(e_0)_{r,s} 
=
\frac1{\sqrt d \|v\|}
e^{\pi i rs/d}
\overline{A_r(g)}.
\]
In particular, for fixed \(r\), varying \(s\) changes only the Weyl
phase:
\[
B_0(e_0)_{r,s} 
=
e^{\pi i rs/d}\,B_0(e_0)_{r,0} ,
\qquad
\abs{B_0(e_0)_{r,s} }
=
\abs{B_0(e_0)_{r,0} }.
\]
Thus, although the canonical fundamental block contains \(d^2\)
coefficients, for the logical basis state \(e_0\) its values are
explicitly determined by the \(d\) periodized samples
\(A_r(g)\) of the chosen Gabor frame. The remaining coefficients on
the full adjoint lattice are then obtained from
\eqref{eq:stabilizer-coefficient-covariance}.

 \subsection{Example: block action of the Hadamard gate}
\label{subsec:hadamard-block-action}

Consider the square GKP qubit \(d=2\). We order the elements of \(K_2\) as
\[
(0,0),(0,1),(1,0),(1,1).
\]
For a fixed logical probe $v=(v_0,v_1)^T\in \mathbb R^2$,
the coefficient-block map is
\[
B_0:\mathcal H_2\longrightarrow\ell^2(K_2),
\qquad
B_0(c)=\left(\frac1{\sqrt 2 \|v\|}\langle c,W_{r,s}v\rangle\right)_{(r,s)\in K_2}.
\]
With
\[
\mathsf W_{r,s}=e^{\pi i rs/2}X^rZ^s,
\]
their matrices are
\[B_0
=
\frac1{\sqrt 2 \|v\|}
\begin{pmatrix}
 v_0& v_1\\
 v_0&- v_1\\
 v_1& v_0\\
i v_1&-i v_0
\end{pmatrix}.
 \qquad \text{and} \qquad B_0^* =
\frac1{\sqrt 2 \|v\|}\begin{pmatrix}
{v_0} & {v_0} &
{v_1} & -i{v_1}\\
{v_1} & -{v_1} &
{v_0} & i{v_0}
\end{pmatrix}.\]

Now  let $
H=\frac1{\sqrt2}
\begin{pmatrix}
1&1\\
1&-1
\end{pmatrix}$  be the Hadamard gate.  Set $u:=v_0+v_1$ and  $w:=v_0-v_1$. Then
\[
B_0HB^*_0
=
\frac{1}{2\sqrt2 \|v\|^2}
\begin{pmatrix}
u{v_0}+w{v_1}
&
u{v_0}-w{v_1}
&
u{v_1}+w{v_0}
&
-i u{v_1}+i w{v_0}
\\[1mm]
w{v_0}+u{v_1}
&
w{v_0}-u{v_1}
&
w{v_1}+u{v_0}
&
-i w{v_1}+i u{v_0}
\\[1mm]
u{v_0}-w{v_1}
&
u{v_0}+w{v_1}
&
u{v_1}-w{v_0}
&
-i u{v_1}-i w{v_0}
\\[1mm]
-iw{v_0}+iu{v_1}
&
-iw{v_0}-iu{v_1}
&
-iw{v_1}+iu{v_0}
&
-w{v_1}-u{v_0}
\end{pmatrix}.
\]
 
Observe that the above matrix depends on the probe vector we choose. So a particularly convenient choice is to take the probe \(v\) to be
an eigenvector of the Hadamard gate. For example, $v=(\cos(\pi/8),\sin(\pi/8))$ is such that $Hv=v$.
Then we have 
\[
B_0H=
\begin{pmatrix}
1&0&0&0\\
0&0&1&0\\
0&1&0&0\\
0&0&0&-1
\end{pmatrix}B_0.
\]
Thus, the Hadamard gate exchanges the coefficients indexed by
\((0,1)\) and \((1,0)\), fixes the coefficient indexed by \((0,0)\),
and changes the sign of the coefficient indexed by \((1,1)\).

\section{Normalizable lattice envelopes and logical coefficient recovery}

Under the general hypotheses used in this section, the
lattice-envelope construction produces states in
$L^2(\mathbb R^n)$ and hence normalizable GKP approximants.
Finite mean oscillator energy is a stronger property, requiring
additional weighted regularity in both position and momentum.
We therefore use ``normalizable'' throughout the general theory.

Ideal GKP codewords are distributional and therefore do not represent normalizable finite-energy oscillator states. This motivates replacing the infinite stabilizer orbit by a suitably decaying lattice envelope. In the time--frequency framework developed here, such regularizations arise naturally as Gabor multipliers: the logical information is encoded in a fixed localized window, while the multiplier symbol controls the decay along the stabilizer lattice; see \cite{FeichtingerNowak,Balazs,FeichtingerHalvdanssonLuef}. In this section we study these normalizable approximations, their convergence to the ideal code, the asymptotic preservation of the logical inner product, and the recovery of logical information from their adjoint-lattice coefficients.

For the rest of the paper we fix a multi-window Gabor frame $\mathbf g\in M^1(\mathbb R^n)\otimes \mathbb C^m$  for the lattice $\Lambda_D$  and its dual $\boldsymbol\gamma$.
\label{sec:regularization}

\subsection{Orbit expansion and logical coordinates}

Let
\[
v_r:=\cZ_Dg_r(0,0)\in\cH_{\mathbf d} \qquad r=1,\ldots,m.
\]
The modulation-space Gabor reconstruction formula applied to the ideal
codeword gives
\begin{equation*}
\psi_c
=\sum_{r=1}^m\sum_{\lambda\in\Lambda_D}
\ip{\psi_c}{\rho(\lambda)g_r}\rho(\lambda)\gamma_r.
\end{equation*}
Zak duality \eqref{eq:vector-zak-duality} gives \(\ip{\psi_c}{g_r}=\langle \mathcal Z^{-1}_D(c\delta_{(0,0)}),g_r\rangle=\langle c\delta_{(0,0)},\mathcal Z_D(g_r)\rangle =\ip{c}{v_r}\). Define
\begin{equation*}
\Gamma_c:=\sum_{r=1}^m\ip{c}{v_r}\gamma_r\in M^1(\R^n).
\end{equation*}
Then
\begin{equation}
\label{eq:ideal-orbit-expansion}
    \psi_c  =\sum_{r=1}^m\sum_{\lambda\in\Lambda_D}
\ip{\rho(-\lambda)\psi_c}{g_r}\rho(\lambda)\gamma_r=\sum_{r=1}^m\sum_{\lambda\in\Lambda_D}
\chi(\lambda)\ip{\psi_c}{g_r}\rho(\lambda)\gamma_r  =\sum_{\lambda\in\Lambda_D}S_\lambda\Gamma_c    
\end{equation}
in the weak-* topology of \(M^\infty\). Applying the vector Zak transform
and Fourier inversion on the syndrome torus gives
\begin{equation}
\label{eq:Gamma-fibre-normalization}
\cZ_D\Gamma_c(0,0)=\sqrt{d}\,c.
\end{equation}

\subsection{Lattice-envelope regularization and Gabor multipliers}
\label{subsec:lattice-envelope-regularization}

The regularization introduced below is naturally a Gabor multiplier.
Giving a scalar symbol
\[
w=(w_\lambda)_{\lambda\in\Lambda_D}\in \ell^\infty(\Lambda_D),
\]
the associated multi-window Gabor multiplier is formally given by
\begin{equation*}
\mathcal M_w^{\boldsymbol\gamma,\mathbf g}f
:=
\sum_{r=1}^m
\sum_{\lambda\in\Lambda_D}
w_\lambda
\ip{f}{\rho(\lambda)g_r}
\rho(\lambda)\gamma_r \qquad f\in M^\infty(\mathbb R^n).
\end{equation*}
Gabor multipliers are standard time--frequency operators; see
\cite{FeichtingerNowak,Balazs}, and see also
\cite{FeichtingerHalvdanssonLuef} for a related operator-convolution
perspective.

We now apply this construction to an ideal codeword
\(\psi_c\in\cC_{\Lambda_D}\).  Since
\[
\ip{\psi_c}{\rho(\lambda)g_r}
=
\chi_D(\lambda)\ip{\psi_c}{g_r}
=
\chi_D(\lambda)\ip{c}{v_r},
\qquad
v_r=\cZ_Dg_r(0,0),
\]
we obtain
\begin{align*}
\mathcal M_w^{\boldsymbol \gamma,\mathbf g}\psi_c
&=
\sum_{r=1}^m
\ip{c}{v_r}
\sum_{\lambda\in\Lambda_D}
w_\lambda\chi_D(\lambda)
\rho(\lambda)\gamma_r
\nonumber\\
&=
\sum_{\lambda\in\Lambda_D}
w_\lambda S_\lambda
\left(
\sum_{r=1}^m\ip{c}{v_r}\gamma_r
\right)
\nonumber\\
&=
\sum_{\lambda\in\Lambda_D}
w_\lambda S_\lambda\Gamma_c .
\end{align*}
Thus, the lattice-envelope regularization is precisely the Gabor
multiplier
\begin{equation}
\label{eq:regularized-codeword}
\Psi_c^w
:=
\mathcal M_w^{\gamma,g}\psi_c
=
\sum_{\lambda\in\Lambda_D}
w_\lambda S_\lambda\Gamma_c .
\end{equation}

This interpretation separates the two roles in the construction.
The function $\Gamma_c
=
\sum_{r=1}^m\ip{c}{v_r}\gamma_r$
contains the logical information, whereas the multiplier symbol
\(w\) controls the envelope of the stabilizer orbit.  The ideal
codeword corresponds formally to the constant symbol $w_\lambda\equiv1$, while decaying symbols produce normalizable approximations.  In
particular, the regularization of a GKP codeword is not an additional
construction external to Gabor analysis: it is the action of a
standard Gabor multiplier on the ideal distributional codeword.

The usual modulation-space analysis and synthesis estimates give
constants \(C_1,C_2>0\), independent of \(c\) and \(w\), such that
\begin{align*}
\norm{\Psi_c^w}_{L^2}
&\le C_2\norm{w}_{\ell^2}\norm{c},
&& w\in\ell^2(\Lambda_D),
\\
\norm{\Psi_c^w}_{M^1}
&\le C_1\norm{w}_{\ell^1}\norm{c},
&& w\in\ell^1(\Lambda_D).
\end{align*}
Indeed, since the coefficients
\(\ip{c}{v_r}\) depend continuously on \(c\), these estimates follow
from the boundedness of the corresponding Gabor synthesis maps.
In particular, every \(\ell^2\)-symbol produces a normalizable state.

For \(\nu\in\Lambda_D\), set
\[
(T_\nu w)_\lambda:=w_{\lambda-\nu}.
\]
The stabilizer action has a simple interpretation at the level of the
multiplier symbol. Since \(\lambda\mapsto S_\lambda\) is a
representation,
\begin{align*}
S_\nu\Psi_c^w
&=
\sum_{\lambda\in\Lambda_D}
w_\lambda S_{\nu+\lambda}\Gamma_c\\
&=
\sum_{\kappa\in\Lambda_D}
w_{\kappa-\nu}S_\kappa\Gamma_c
=
\Psi_c^{T_\nu w}.
\end{align*}
Hence $S_\nu\Psi_c^w=\Psi_c^{T_\nu w}$, and consequently
\begin{equation}
\label{eq:stabilizer-defect}
\norm{S_\nu\Psi_c^w-\Psi_c^w}_2
\le
C_2
\norm{T_\nu w-w}_{\ell^2}
\norm{c}.
\end{equation}
Thus, approximate stabilization is controlled precisely by the
translation invariance of the Gabor-multiplier symbol.

If \(\{w^{(\varepsilon)}\}_{\varepsilon>0}\) is uniformly bounded in
\(\ell^\infty(\Lambda_D)\) and
\[
w^{(\varepsilon)}
\stackrel{w^*}{\longrightarrow}
\mathbf1
\quad\text{in }\ell^\infty(\Lambda_D),
\]
then
\begin{equation}
\label{eq:weak-star-ideal-limit}
\Psi_c^{w^{(\varepsilon)}}
\stackrel{w^*}{\longrightarrow}
\psi_c
\quad\text{in }M^\infty(\R^n).
\end{equation}
Indeed, let \(\varphi\in M^1(\R^n)\). Since
\(\Gamma_c,\varphi\in M^1(\R^n)\), then $\left(
\ip{S_\lambda\Gamma_c}{\varphi}
\right)_{\lambda\in\Lambda_D}$
belongs to \(\ell^1(\Lambda_D)\). Therefore
\begin{align*}
\ip{\Psi_c^{w^{(\varepsilon)}}}{\varphi}
&=
\sum_{\lambda\in\Lambda_D}
w_\lambda^{(\varepsilon)}
\ip{S_\lambda\Gamma_c}{\varphi}
\longrightarrow
\sum_{\lambda\in\Lambda_D}
\ip{S_\lambda\Gamma_c}{\varphi}
=
\ip{\psi_c}{\varphi},
\end{align*}
where the last equality follows from
\eqref{eq:ideal-orbit-expansion}.

The Gaussian envelopes
\begin{equation*} \label{eq:gaussian-envelope}
w_\lambda^{(\varepsilon)}=e^{-\varepsilon|\lambda|^2}
    \end{equation*}
belong to \(\ell^1(\Lambda_D)\cap\ell^2(\Lambda_D)\), are uniformly
bounded in \(\ell^\infty\), and converge pointwise to \(1\). Hence
\[
\Psi_c^{w^{(\varepsilon)}}\overset{w^*}{\longrightarrow}\psi_c
\quad\text{in }M^\infty(\mathbb R^n).
\]

\subsection{Asymptotically isometric normalizable encoding}
\label{subsec:asymptotic-isometry}

Weak-* convergence alone does not control the norm geometry of the
normalizable states. The next theorem provides such control.

\begin{definition}[Asymptotically translation-invariant envelope]
\label{def:folner-envelope}
A family of nonzero envelopes \(w^{(\varepsilon)}\in\ell^2(\Lambda_D)\) is
called \emph{asymptotically translation invariant} if, for every fixed
\(\nu\in\Lambda_D\),
\begin{equation}
\label{eq:folner-condition}
\frac{\norm{T_\nu w^{(\varepsilon)}-w^{(\varepsilon)}}_{\ell^2}}
{\norm{w^{(\varepsilon)}}_{\ell^2}}
\longrightarrow0.
\end{equation}
\end{definition}

\begin{lemma}
The Gaussian lattice envelopes
\[
w_\lambda^{(\varepsilon)}
=
e^{-\varepsilon|\lambda|^2},
\qquad \lambda\in\Lambda_D,
\]
are asymptotically translation-invariant.
\end{lemma}

\begin{proof}
Fix $\nu\in\Lambda_D$. A completed-square calculation gives
\[
\frac{
\langle w^{(\varepsilon)},
T_{-\nu}w^{(\varepsilon)}\rangle
}{
\|w^{(\varepsilon)}\|_2^2
}
=
e^{-\varepsilon|\nu|^2/2}
\frac{
\displaystyle
\sum_{\lambda\in\Lambda_D}
e^{-2\varepsilon|\lambda+\nu/2|^2}
}{
\displaystyle
\sum_{\lambda\in\Lambda_D}
e^{-2\varepsilon|\lambda|^2}
}.
\]
By Poisson summation, the quotient of theta sums converges to
$1$ as $\varepsilon\to 0$. Therefore
\[
\frac{
\langle w^{(\varepsilon)},
T_{-\nu}w^{(\varepsilon)}\rangle
}{
\|w^{(\varepsilon)}\|_2^2
}
\longrightarrow1.
\]
Since translations preserve the $\ell^2$-norm,
\[\frac{
\|T_{-\nu}w^{(\varepsilon)}
-w^{(\varepsilon)}\|_2^2
}{
\|w^{(\varepsilon)}\|_2^2
}
=
2-
2\operatorname{Re}
\frac{
\langle w^{(\varepsilon)},
T_{-\nu}w^{(\varepsilon)}\rangle
}{
\|w^{(\varepsilon)}\|_2^2
}
\longrightarrow0.
\]
\end{proof}

\begin{theorem}[Asymptotic isometry]
\label{thm:asymptotic-isometry}
Let \(w^{(\varepsilon)}\) satisfy \eqref{eq:folner-condition}. Define
\begin{equation*}
E_\varepsilon:\mathcal H_{\mathbf d} \to L^2(\mathbb R^n),\qquad c\longmapsto  E_\varepsilon c
:=d^{-1/4}
\frac{\Psi_c^{w^{(\varepsilon)}}}
{\norm{w^{(\varepsilon)}}_{\ell^2}}.
\end{equation*}
Then
\begin{equation*}
E_\varepsilon^*E_\varepsilon
\longrightarrow I_{\cH_{\mathbf d}}
\end{equation*}
in operator norm. Equivalently,
\begin{equation}
\label{eq:asymptotic-gram}
\frac{\ip{\Psi_c^{w^{(\varepsilon)}}}
{\Psi_{c'}^{w^{(\varepsilon)}}}}
{\sqrt{d}\,
\norm{w^{(\varepsilon)}}_{\ell^2}^2}
\longrightarrow\ip{c}{c'}
\end{equation}
uniformly for \(c,c'\) in bounded subsets of \(\cH_{\mathbf d}\).
\end{theorem}

\begin{proof}
For \(w\in\ell^2(\Lambda_D)\), set
\[
R_w(\nu)
:=\sum_{\lambda\in\Lambda_D}
w_\lambda\overline{w_{\lambda+\nu}}.
\]
Expanding \eqref{eq:regularized-codeword} and using the unitarity of the
stabilizer representation gives
\begin{equation*}
\ip{\Psi_c^w}{\Psi_{c'}^w}
=\sum_{\nu\in\Lambda_D}
R_w(\nu)\ip{\Gamma_c}{S_\nu\Gamma_{c'}}.
\end{equation*}
By Cauchy--Schwarz,
\(\abs{R_w(\nu)}\le\norm{w}_{\ell^2}^2\). Moreover,
\[
\frac{\abs{R_{w^{(\varepsilon)}}(\nu)
-\norm{w^{(\varepsilon)}}_{\ell^2}^2}}
{\norm{w^{(\varepsilon)}}_{\ell^2}^2}
\le
\frac{\norm{T_{-\nu}w^{(\varepsilon)}-w^{(\varepsilon)}}_{\ell^2}}
{\norm{w^{(\varepsilon)}}_{\ell^2}},
\]
so
\(R_{w^{(\varepsilon)}}(\nu)/\norm{w^{(\varepsilon)}}_2^2\to1\)
for every fixed \(\nu\).

Because \(\Gamma_c,\Gamma_{c'}\in M^1(\mathbb R^n)\), the sequence $\big(\ip{\Gamma_c}{S_\nu\Gamma_{c'}}\big)_{\nu\in\Lambda_D}$ 
belongs to \(\ell^1(\Lambda_D)\), with an \(\ell^1\)-bound proportional to
\(\norm{c}\norm{c'}\). Dominated convergence yields
\[
\frac{\ip{\Psi_c^{w^{(\varepsilon)}}}
{\Psi_{c'}^{w^{(\varepsilon)}}}}
{\norm{w^{(\varepsilon)}}_2^2}
\longrightarrow
\sum_{\nu\in\Lambda_D}
\ip{\Gamma_c}{S_\nu\Gamma_{c'}}.
\]
By \eqref{eq:ideal-orbit-expansion}, the last sum equals
\(\ip{\Gamma_c}{\psi_{c'}}\). Zak duality \eqref{eq:vector-zak-duality} and
\eqref{eq:Gamma-fibre-normalization} give
\[
\ip{\Gamma_c}{\psi_{c'}}
=\sqrt{d}\,\ip{c}{c'}.
\]
This proves \eqref{eq:asymptotic-gram}. Since the logical space is
finite-dimensional, convergence of the associated sesquilinear forms is
equivalent to operator-norm convergence.
\end{proof}
Theorem~\ref{thm:asymptotic-isometry} should be compared with
\cite[Theorem~5.2]{MayrandRoyer}. There, the approximation is
defined using a canonical number-operator envelope and a
quantitative estimate for the recovery of the logical inner
product is obtained. Our construction instead uses
lattice-symbol Gabor multipliers and applies to general
asymptotically translation-invariant envelopes. The advantage is
the flexibility of the regularizing symbol and its direct
compatibility with the time--frequency coefficient model; the
present conclusion is qualitative operator-norm convergence
rather than a quantitative error estimate. We do not claim that
the two regularization procedures coincide at finite parameter.\begin{corollary}[Normalized stabilizer defect]
\label{cor:normalized-stabilizer-defect}
Under the hypotheses of Theorem~\ref{thm:asymptotic-isometry}, for every
fixed \(\nu\in\Lambda_D\) and every nonzero \(c\),
\begin{equation*}
\frac{\norm{S_\nu\Psi_c^{w^{(\varepsilon)}}
-\Psi_c^{w^{(\varepsilon)}}}_2}
{\norm{\Psi_c^{w^{(\varepsilon)}}}_2}
\longrightarrow0.
\end{equation*}
\end{corollary}

\begin{proof}
Combine \eqref{eq:stabilizer-defect}, the condition
\eqref{eq:folner-condition}, and the norm asymptotics from
\eqref{eq:asymptotic-gram}.
\end{proof}

\subsection{Recovery from regularized coefficients}
\label{subsec:regularized-recovery}

Fix a probe $\alpha\in \mathbb C^m$ with \(v_\alpha\neq0\). Define  the finite-block analysis map  $B_0^{\omega^{(\varepsilon)}}: \mathcal H_{\mathbf d}\to \ell^2(K_D)$ by 
\begin{align*}
    (B_0^{\omega^{(\varepsilon)}}(c))_q & = \frac1{\sqrt d \|v_\alpha\|} \langle  \Psi_c^{w^{(\varepsilon)}}, \rho(\nu_q) h_\alpha\rangle & \qquad q\in K_D,
\end{align*}

If \(w^{(\varepsilon)}\stackrel{w^*}{\to}\mathbf1\) in
\(\ell^\infty(\Lambda_D)\) and is uniformly bounded, then
\begin{equation*}
\delta_\varepsilon
:=
\norm{B_0^{\omega^{(\varepsilon)}}-B_0}
\longrightarrow0.
\end{equation*}
Indeed, by \eqref{eq:scalar-probe-coeff},
\[
(B_0^{\omega^{(\varepsilon)}} (c))_q
=
\frac1{\sqrt d \|v_\alpha\|}\ip{\Psi_c^{w^{(\varepsilon)}}}
{\rho(\nu_q)h_\alpha} \qquad q\in K_D.
\]
Since \(\rho(\nu_q)h_\alpha\in M^1(\R^n)\) for every  \(q\in K_D\), the weak-* convergence
\eqref{eq:weak-star-ideal-limit} gives, for every fixed
\(c\in\cH_{\mathbf d}\),
\[
(B_0^{\omega^{(\varepsilon)}} (c))_q
\longrightarrow
\frac1{\sqrt d \|v_\alpha\|}
\ip{\psi_c}{\rho(\nu_q)h_\alpha}
=
\frac1{\sqrt d \|v_\alpha\|}
\ip{c}{\mathsf W_qv_\alpha}
=
(B_0(c))_q.
\]
Since \(K_D\) is finite, it follows that
\[
\norm{( B_0^{\omega^{(\varepsilon)}}-B_0)(c)}_{\ell^2(K_D)}
\longrightarrow0
\]
for every \(c\in\cH_{\mathbf d}\). Finally,
\(\cH_{\mathbf d}\) is finite-dimensional, so pointwise convergence
of the linear maps \(B_0^{\omega^{(\varepsilon)}}-B_0\) implies convergence in
operator norm. Hence \(\delta_\varepsilon\to0\).

Theorem~\ref{thm:finite-block-recovery} gives
\begin{equation*}
\norm{B_0(c)}_{\ell^2(K_D)}
=\norm{c},
\end{equation*}
and therefore
\begin{equation*}
\left(
1-\delta_\varepsilon
\right)\norm{c}
\le
\norm{B_0^{\omega^{(\varepsilon)}} (c)}
\le
\left(
1+\delta_\varepsilon
\right)\norm{c}.
\end{equation*}
Thus, \(B_0^{\omega^{(\varepsilon)}}\) is injective whenever
\begin{equation*}
\delta_\varepsilon
<1.
\end{equation*}
For noisy data
\(\widetilde b^{(\varepsilon)}=B_0^{\omega^{(\varepsilon)}}(c)+e\), the least-squares left
inverse
\[
R_\varepsilon
=((B_0^{\omega^{(\varepsilon)}})^*B_0^{\omega^{(\varepsilon)}})^{-1} (B_0^{\omega^{(\varepsilon)}})^*
\]
satisfies
\begin{equation*}
\norm{R_\varepsilon\widetilde b^{(\varepsilon)}-c}
\le
\frac{\norm{e}_{\ell^2(K_D)}}
{1-\delta_\varepsilon}.
\end{equation*}

The corresponding orthogonal projection onto the space of admissible
regularized blocks is
\[
    \Pi_\varepsilon
    :=
    B_0^{\omega^{(\varepsilon)}}
    \bigl(
       (B_0^{\omega^{(\varepsilon)}})^*
       B_0^{\omega^{(\varepsilon)}}
    \bigr)^{-1}
    (B_0^{\omega^{(\varepsilon)}})^*
    =
    B_0^{\omega^{(\varepsilon)}}R_\varepsilon .
\]
Hence, for arbitrary data $b\in\ell^2(K_D)$,
\[
    \Pi_\varepsilon b
    =
    \operatorname*{argmin}_{
        a\in\operatorname{Ran}(B_0^{\omega^{(\varepsilon)}})}
    \|b-a\|_{\ell^2(K_D)},
\]
while $R_\varepsilon b$ is the logical vector associated with this
nearest regularized block.
\section{Syndrome information in coefficient space}
\label{sec:syndrome}

The usual GKP syndrome is encoded in the phases of stabilizer
eigenvalues; in lattice coordinates this amounts to recovering the
displacement modulo the adjoint lattice from a family of symplectic
characters, see for instance
\cite{ConradEisertArzani,TerhalConradVuillot}.
Continuous or ``analog'' GKP syndrome information has also been used
to improve decoding through likelihood-based methods; see
\cite{FukuiTomitaOkamoto,NohChamberland}.

Our purpose here is different but closely related: we show that the
same syndrome character can be read directly from the
adjoint-lattice Gabor coefficients. Proposition~\ref{prop:block-syndrome-covariance}
expresses the syndrome as a phase relation between translated finite
coefficient blocks. We then extract this phase by block correlation,
quantify its stability under coefficient noise, and exploit redundant
stabilizer directions for oversampled syndrome estimation. Finally, we
show that these relations persist asymptotically for the normalizable
regularizations of Section~\ref{sec:regularization}.

\subsection{Block syndrome covariance}

Let \(f\in\cC_{\Lambda_D}\) and a probe $\alpha\in \mathbb C^m$. Then 
the displaced state \(\rho(\zeta)f\) satisfies
\begin{equation*}
\rho(-\lambda)\rho(\zeta)f
=
\chi_D(-\lambda)
e^{2\pi i\sigma(\lambda,\zeta)}
\rho(\zeta)f.
\end{equation*}
Thus, a displacement error converts the stabilizer eigenvalue into the
character
\[
\lambda
\longmapsto
e^{2\pi i\sigma(\lambda,\zeta)}\qquad \lambda\in \Lambda_D.
\]
For \(\lambda\in\Lambda_D\), define the translated and phase-corrected
coefficient block
\begin{equation*}
B_\lambda^{(\zeta)}(f)
:=
\frac{\sqrt d}{\|v_\alpha\|}\left(
e^{\pi i\sigma(\lambda,\nu_q)}
a_{\lambda+\nu_q}
\bigl(
\Phi_{\gamma,\alpha}(\rho(\zeta)f)
\bigr)
\right)_{q\in K_D}
\in\ell^2(K_D).
\end{equation*}
Observe that $B^{(0)}_\lambda(\psi_c)$ is by definition $B_\lambda(c)$ from Corollary~\ref{cor:periodicity-finite-blocks}.

\begin{proposition}[Block syndrome covariance]
\label{prop:block-syndrome-covariance}
Let \(f\in\mathcal C_{\Lambda_D}\),
\(\zeta\in\mathbb R^{2n}\), and
\(\lambda\in\Lambda_D\). Then
\begin{equation}
\label{eq:block-syndrome-covariance}
B_\lambda^{(\zeta)}(f)
=
s_\lambda(\zeta)B_0^{(\zeta)}(f),
\end{equation}
where
\begin{equation}
\label{eq:syndrome-character}
s_\lambda(\zeta)
:=
\chi_D(\lambda)e^{2\pi i\sigma(\lambda,\zeta)}.
\end{equation}
Equivalently, for every \(\mu\in\Lambda_D^\circ\),
\begin{equation}
\label{eq:coefficient-syndrome-covariance}
a_{\mu+\lambda}
\bigl(\Phi_{\gamma,\alpha}(\rho(\zeta)f)\bigr)
=
\chi_D(\lambda)e^{2\pi i\sigma(\lambda,\zeta)}
e^{-\pi i\sigma(\lambda,\mu)}
a_\mu
\bigl(\Phi_{\gamma,\alpha}(\rho(\zeta)f)\bigr).
\end{equation}
\end{proposition}

\begin{proof}
Set $F_\zeta:=\rho(\zeta)f$. By the coefficient formula,
\begin{align*}
a_{\mu+\lambda}
\bigl(\Phi_{\gamma,\alpha}(F_\zeta)\bigr)
& =
\frac{1}{d}
\left\langle
F_\zeta,\rho(\mu+\lambda)h_\alpha
\right\rangle  \\ 
    &=
\frac{e^{-\pi i\sigma(\lambda,\mu)}}{d}
\left\langle
F_\zeta,\rho(\lambda)\rho(\mu)h_\alpha
\right\rangle
\\
&=
\frac{e^{-\pi i\sigma(\lambda,\mu)}}{d}
\left\langle
\rho(-\lambda)F_\zeta,\rho(\mu)h_\alpha
\right\rangle .
\end{align*}
Then using  the stabilizer condition we have
\begin{align*}
\rho(-\lambda)F_\zeta
&=
\rho(-\lambda)\rho(\zeta)f
\\
&=
e^{2\pi i\sigma(\lambda,\zeta)}
\rho(\zeta)\rho(-\lambda)f
\\
&=
\chi_D(\lambda)e^{2\pi i\sigma(\lambda,\zeta)}
F_\zeta.
\end{align*}
Substitution yields
\[
a_{\mu+\lambda}
\bigl(\Phi_{\gamma,\alpha}(F_\zeta)\bigr)
=
\chi_D(\lambda)e^{2\pi i\sigma(\lambda,\zeta)}
e^{-\pi i\sigma(\lambda,\mu)}
a_\mu
\bigl(\Phi_{\gamma,\alpha}(F_\zeta)\bigr),
\]
which proves \eqref{eq:coefficient-syndrome-covariance}.

Taking \(\mu=\nu_q\) and multiplying by
\(e^{\pi i\sigma(\lambda,\nu_q)}\), we obtain
\[
e^{\pi i\sigma(\lambda,\nu_q)}
a_{\lambda+\nu_q}
\bigl(\Phi_{\gamma,\alpha}(F_\zeta)\bigr)
=
s_\lambda(\zeta)
a_{\nu_q}
\bigl(\Phi_{\gamma,\alpha}(F_\zeta)\bigr).
\]
This holds for every \(q\in K_D\), and hence
\[
B_\lambda^{(\zeta)}(f)
=
s_\lambda(\zeta)B_0^{(\zeta)}(f).
\]
\end{proof}

\begin{remark}
Taking \(\mu=0\) in
\eqref{eq:coefficient-syndrome-covariance} gives the scalar relation
\[
\frac{
a_\lambda
\bigl(\Phi_{\gamma,\alpha}(\rho(\zeta)f)\bigr)
}{
a_0
\bigl(\Phi_{\gamma,\alpha}(\rho(\zeta)f)\bigr)
}
=
\chi_D(-\lambda)
e^{2\pi i\sigma(\lambda,\zeta)}
\]
whenever the denominator is nonzero. Proposition~
\ref{prop:block-syndrome-covariance} strengthens this observation by
replacing a single reference coefficient with the entire finite block.
\end{remark}

\subsection{Block-correlation extraction of the syndrome}

The covariance relation
\eqref{eq:block-syndrome-covariance} allows the syndrome character to
be extracted by correlating two finite coefficient blocks.

\begin{corollary}[Block-correlation syndrome formula]
\label{cor:block-correlation-syndrome}
Let \(f\in\cC_{\Lambda_D}\), and assume that
\(B_0^{(\zeta)}(f)\neq0\). Then, for every
\(\lambda\in\Lambda_D\),
\begin{equation}
\label{eq:block-correlation-formula}
\frac{
\ip{B_\lambda^{(\zeta)}(f)}
{B_0^{(\zeta)}(f)}
}{
\norm{B_0^{(\zeta)}(f)}_{\ell^2(K_D)}^2
}
=
\chi_D(\lambda)
e^{2\pi i\sigma(\lambda,\zeta)}.
\end{equation}

In particular, if \(f=\psi_c\neq 0 \), with  
\(v_\alpha=G(0,0)\overline\alpha\neq0\), then $B_0^{(\zeta)}(\psi_c)\neq0$  for all sufficiently small
displacements \(\zeta\).
\end{corollary}

\begin{proof}
By Proposition~\ref{prop:block-syndrome-covariance},
\[
B_\lambda^{(\zeta)}(f)
=
s_\lambda(\zeta)B_0^{(\zeta)}(f).
\]
Since our inner product is linear in the first variable,
\begin{align*}
\ip{B_\lambda^{(\zeta)}(f)}
{B_0^{(\zeta)}(f)}
&=
s_\lambda(\zeta)
\ip{B_0^{(\zeta)}(f)}
{B_0^{(\zeta)}(f)}
\\
&=
s_\lambda(\zeta)
\norm{B_0^{(\zeta)}(f)}_{\ell^2(K_D)}^2.
\end{align*}
Dividing by the nonzero denominator and using
\eqref{eq:syndrome-character} gives
\eqref{eq:block-correlation-formula}.

Now let \(f=\psi_c\neq 0\). At zero displacement,
\(B_0^{(0)}(\psi_c)\) is precisely the finite block appearing in
Theorem~\ref{thm:finite-block-recovery}, so \(B_0^{(0)}(\psi_c)\neq0\). 
Finally, for each fixed \(q\in K_D\),
\[
\zeta
\longmapsto
a_{\nu_q}
\bigl(
\Phi_{\gamma,\alpha}(\rho(\zeta)\psi_c)
\bigr)
\]
is continuous. Indeed, by \eqref{eq:scalar-probe-coeff} it is a
distributional pairing of \(\psi_c\) against a continuously translated
\(M^1\)-window. Since \(K_D\) is finite,
\[
\zeta\longmapsto B_0^{(\zeta)}(\psi_c)
\]
is continuous as an \(\ell^2(K_D)\)-valued map. Therefore its norm
remains strictly positive in a neighbourhood of \(\zeta=0\).
\end{proof}

The advantage of \eqref{eq:block-correlation-formula} over the scalar
ratio is that no individual coefficient needs to be nonzero. 

\subsection{Stability under coefficient noise}

The redundancy of the finite Weyl block also gives a natural
estimator with a deterministic coefficient-perturbation bound of the syndrome character.

Suppose that the exact blocks satisfy 
$B_\lambda=sB_0$,
 with $|s|=1$, and that the available complex coefficients blocks are
\begin{equation*}
\widetilde B_0=B_0+e_0,
\qquad
\widetilde B_\lambda=sB_0+e_\lambda, \qquad \text{for }e_0,e_\lambda\in \ell^2(K_D).
\end{equation*}
Define the block-correlation estimator
\begin{equation}
\label{eq:noisy-syndrome-estimator}
\widehat s
:=
\frac{
\ip{\widetilde B_\lambda}{\widetilde B_0}
}{
\norm{\widetilde B_0}^2
}.
\end{equation}

\begin{proposition}[Coefficient-noise stability]
\label{prop:block-correlation-noise}
Assume
\begin{equation}
\label{eq:relative-block-noise}
\norm{e_0}
\leq
\eta\norm{B_0},
\qquad
\norm{e_\lambda}
\leq
\eta\norm{B_0},
\qquad
0\leq\eta<1.
\end{equation}
Then \(\widetilde B_0\neq0\), and the estimator
\eqref{eq:noisy-syndrome-estimator} satisfies
\begin{equation}
\label{eq:block-correlation-noise-bound}
\abs{\widehat s-s}
\leq
\frac{2\eta}{1-\eta}.
\end{equation}
\end{proposition}

\begin{proof}
By the reverse triangle inequality and
\eqref{eq:relative-block-noise},
\[
\norm{\widetilde B_0}
\geq
\norm{B_0}-\norm{e_0}
\geq
(1-\eta)\norm{B_0}>0.
\]
Hence \(\widehat s\) is well defined.

Using \(\widetilde B_\lambda=sB_0+e_\lambda\) and
\(\widetilde B_0=B_0+e_0\), we obtain
\begin{align*}
\ip{\widetilde B_\lambda}{\widetilde B_0}
-s\norm{\widetilde B_0}^2
&=
s\ip{B_0}{\widetilde B_0}
+\ip{e_\lambda}{\widetilde B_0}
-s\ip{\widetilde B_0}{\widetilde B_0}
\\
&=
\ip{e_\lambda}{\widetilde B_0}
-s\ip{e_0}{\widetilde B_0}.
\end{align*}
Therefore, by Cauchy--Schwarz,
\begin{align*}
\abs{\widehat s-s}
&\leq
\frac{
\bigl(
\norm{e_\lambda}+\norm{e_0}
\bigr)
\norm{\widetilde B_0}
}{
\norm{\widetilde B_0}^2
}
\\
&=
\frac{
\norm{e_\lambda}+\norm{e_0}
}{
\norm{\widetilde B_0}
}
\\
&\leq
\frac{
2\eta\norm{B_0}
}{
(1-\eta)\norm{B_0}
}
=
\frac{2\eta}{1-\eta}.
\end{align*}
This proves \eqref{eq:block-correlation-noise-bound}.
\end{proof}

\subsection{Redundant syndrome estimation}
\label{subsec:redundant-syndrome-estimation}

We now combine the character values obtained from different stabilizer
directions in order to recover the complete displacement syndrome.
The estimation of continuous displacements from modular or stabilizer
phases has appeared previously in GKP quantum sensing and error
correction. Grid states were used as probes for estimating both
components of a phase-space displacement through phase-estimation
measurements in \cite{DuivenvoordenTerhalWeigand}. Analog GKP syndrome
values have also been incorporated into maximum-likelihood and
Bayesian decoders, including schemes based on repeated noisy syndrome
extraction; see
\cite{FukuiTomitaOkamoto,VuillotAsasiWangPryadkoTerhal,
WanNevilleKolthammer}. Our construction is different: the syndrome
phases are extracted from correlations of translated
adjoint-lattice Gabor coefficient blocks, and redundancy is introduced
by using several stabilizer directions in the lattice.

Recall that a displacement \(\zeta\in\mathbb R^{2n}\) gives rise to
the phase
\[
s_\lambda(\zeta)
=
\chi_D(\lambda)e^{2\pi i\sigma(\lambda,\zeta)},
\qquad
\lambda\in\Lambda_D.
\]
Since the stabilizer phase \(\chi_D\) is known, the relevant syndrome
information is the character
\[
\lambda
\longmapsto
\overline{\chi_D(\lambda)}s_\lambda(\zeta)
=
e^{2\pi i\sigma(\lambda,\zeta)},\qquad \lambda\in \Lambda_D.
\]
This character determines precisely the class of \(\zeta\) modulo the
adjoint lattice. Indeed, two displacements \(\zeta,\zeta'\in
\mathbb R^{2n}\) determine the same character if and only if
\[
e^{2\pi i\sigma(\lambda,\zeta-\zeta')}=1
\qquad
\text{for every }\lambda\in\Lambda_D.
\]
Equivalently,
\[
\sigma(\lambda,\zeta-\zeta')\in\mathbb Z
\qquad
\text{for every }\lambda\in\Lambda_D,
\]
which, by the definition of the adjoint lattice, is equivalent to
\[
\zeta-\zeta'\in\Lambda_D^\circ.
\]
Thus the syndrome character determines exactly the point
\[
[\zeta]\in\mathbb R^{2n}/\Lambda_D^\circ.
\]

Let \(\lambda_1,\ldots,\lambda_{2n}\) be a
\(\mathbb Z\)-basis of \(\Lambda_D\). Since a character of
\(\Lambda_D\) is determined by its values on a lattice basis, the
\(2n\) phases
\[
s_{\lambda_j}(\zeta),
\qquad
j=1,\ldots,2n,
\]
are sufficient to recover \([\zeta]\) in the noiseless setting. More
generally, we may choose
\[
\lambda_1,\ldots,\lambda_N\in\Lambda_D,
\qquad
N\geq 2n,
\]
provided that these vectors generate \(\Lambda_D\) as an abelian
group. The additional stabilizer directions do not provide new
information in the noiseless case, but they give redundant phase
measurements that can improve robustness in the presence of
coefficient noise.

For each \(\lambda_j\), let
\[
\widehat s_{\lambda_j}
:=
\frac{
\langle
\widetilde B_{\lambda_j},
\widetilde B_0
\rangle
}{
\|\widetilde B_0\|_{\ell^2(K_D)}^2
}
\]
be the block-correlation estimate obtained from the observed
coefficient blocks, as in \eqref{eq:noisy-syndrome-estimator}.
The predicted phase associated with a candidate syndrome
\([z]\in\mathbb R^{2n}/\Lambda_D^\circ\) is
\[
s_{\lambda_j}(z)
=
\chi_D(\lambda_j)e^{2\pi i\sigma(\lambda_j,z)}.
\]
We therefore define the oversampled syndrome estimator by
\begin{equation}
\label{eq:oversampled-syndrome-estimator}
\widehat\zeta
\in
\operatorname*{argmin}_{[z]\in\mathbb R^{2n}/\Lambda_D^\circ}
\sum_{j=1}^N
\left|
\widehat s_{\lambda_j}
-
\chi_D(\lambda_j)e^{2\pi i\sigma(\lambda_j,z)}
\right|^2.
\end{equation}
The objective is well defined on the quotient because, for
\(\nu\in\Lambda_D^\circ\),
\[
e^{2\pi i\sigma(\lambda_j,z+\nu)}
=
e^{2\pi i\sigma(\lambda_j,z)}
e^{2\pi i\sigma(\lambda_j,\nu)}
=
e^{2\pi i\sigma(\lambda_j,z)}.
\]
Moreover, the syndrome torus
\(\mathbb R^{2n}/\Lambda_D^\circ\) is compact, so a minimizer in
\eqref{eq:oversampled-syndrome-estimator} always exists.

In the absence of coefficient noise,
\[
\widehat s_{\lambda_j}
=
s_{\lambda_j}(\zeta),
\qquad
j=1,\ldots,N.
\]
Hence the objective in
\eqref{eq:oversampled-syndrome-estimator} vanishes at \([\zeta]\).
If the vectors \(\lambda_1,\ldots,\lambda_N\) generate
\(\Lambda_D\), then any other zero \([z]\) satisfies
\[
e^{2\pi i\sigma(\lambda_j,z-\zeta)}=1
\qquad
\text{for every }j.
\]
The same relation then holds for every \(\lambda\in\Lambda_D\), and
therefore $z-\zeta\in\Lambda_D^\circ$. Thus, \([\zeta]\) is the unique minimizer in the noiseless setting.

The redundancy has a concrete interpretation. If, for example,
\[
\lambda_k
=
m_1\lambda_1+\cdots+m_{2n}\lambda_{2n},
\qquad
m_j\in\mathbb Z,
\]
then the ideal character values satisfy
\[
e^{2\pi i\sigma(\lambda_k,\zeta)}
=
\prod_{j=1}^{2n}
e^{2\pi i m_j\sigma(\lambda_j,\zeta)}.
\]
Therefore the measurement associated with \(\lambda_k\) is determined
by those associated with a lattice basis in the noiseless case.
Under coefficient noise, however, the corresponding estimated phases
will not satisfy this relation exactly. Including such additional
directions produces an overdetermined system that provides redundancy and allows weighted or robust fitting.

\subsection{Regularized states}
\label{subsec:regularized-syndrome}

Finite-energy GKP states are not exact eigenvectors of the ideal
stabilizers, and several approaches have been developed to describe
their approximate or modified stabilizer structure and their relation
to ideal grid states; see
\cite{TzitrinBourassaMenicucciSabapathy,
MatsuuraYamasakiKoashi,RoyerSinghGirvin}.
Our purpose here is to quantify the corresponding defect directly in
the adjoint-lattice coefficient model. We show that the failure of
exact block covariance is controlled by the translation defect of the
lattice envelope and that, under weak-* approximation of the ideal
codeword, the block-correlation estimator converges to the ideal
syndrome character.

For \(w\in\ell^2(\Lambda_D)\), define
\begin{equation*}
B_\lambda^{(\zeta),w}(c)
:=
\frac{\sqrt d }{\|v_\alpha\|}\left(
e^{\pi i\sigma(\lambda,\nu_q)}
a_{\lambda+\nu_q}
\bigl(
\Phi_{\gamma,\alpha}(\rho(\zeta)\Psi_c^w)
\bigr)
\right)_{q\in K_D}, \qquad c\in \mathcal H_{\mathbf d}.
\end{equation*}

For a general state \(f\in M^\infty(\mathbb R^n)\), the same calculation
as in Proposition~\ref{prop:block-syndrome-covariance} gives
\begin{equation*}
a_{\mu+\lambda}\bigl(\Phi_{\gamma,\alpha}(f)\bigr)
=
e^{-\pi i\sigma(\lambda,\mu)}
a_\mu\bigl(\Phi_{\gamma,\alpha}(\rho(-\lambda)f)\bigr).
\end{equation*}

Taking \(f=\rho(\zeta)\Psi_c^w\) and using \eqref{eq:weyl-commutation_0} we obtain
\begin{align*}
&
a_{\mu+\lambda}
\bigl(
\Phi_{\gamma,\alpha}(\rho(\zeta)\Psi_c^w)
\bigr)
-
\chi_D(-\lambda)e^{2\pi i\sigma(\lambda,\zeta)}
e^{-\pi i\sigma(\lambda,\mu)}
a_\mu
\bigl(
\Phi_{\gamma,\alpha}(\rho(\zeta)\Psi_c^w)
\bigr)
\nonumber\\
&\qquad
=
e^{-\pi i\sigma(\lambda,\mu)}
e^{2\pi i\sigma(\lambda,\zeta)}
a_\mu
\left(
\Phi_{\gamma,\alpha}
\left(
\rho(\zeta)
\bigl[
\rho(-\lambda)\Psi_c^w
-
\chi_D(-\lambda)\Psi_c^w
\bigr]
\right)
\right).
\end{align*}

Then by (\ref{eq:coeff}) and using 
Cauchy--Schwarz, we have that 
\begin{align}
\label{eq:regularized-block-defect-stabilizer}
\left\|
B_\lambda^{(\zeta),w}(c)
-
s_\lambda(\zeta)B_0^{(\zeta),w}(c)
\right\|_{\ell^2(K_D)}
\leq\frac{\sqrt d }{\|v_\alpha\|}
\|h_\alpha\|_2
\left\|
\rho(-\lambda)\Psi_c^w
-
\chi_D(-\lambda)\Psi_c^w
\right\|_2.
\end{align}
For the canonical phase \(\chi_D\), one has
\(\chi_D(-\lambda)\in\{\pm1\}\), then  
\[
\begin{aligned}
\left\|
\rho(-\lambda)\Psi_c^w
-\chi_D(-\lambda)\Psi_c^w
\right\|_2
&=
\left\|
\chi_D(-\lambda)\rho(-\lambda)\Psi_c^w
-\Psi_c^w
\right\|_2\\
&=
\left\|
S_{-\lambda}\Psi_c^w-\Psi_c^w
\right\|_2.
\end{aligned}
\]
Applying \eqref{eq:stabilizer-defect} with \(\nu=-\lambda\) to
\eqref{eq:regularized-block-defect-stabilizer} therefore gives
\[
\left\|
B_\lambda^{(\zeta),w}(c)
-s_\lambda(\zeta)B_0^{(\zeta),w}(c)
\right\|_{\ell^2(K_D)}
\leq
C
\|T_{-\lambda}w-w\|_{\ell^2}
\|c\|,
\]
for some $C>0$. We first apply this estimate to the the asymptotically isometric normalizable encodings. Let \(w^{(\varepsilon)}\) be asymptotically
translation-invariant in the sense of
Definition~\ref{def:folner-envelope}, and recall
that
\[
E_\varepsilon c
=
d^{-1/4}
\frac{\Psi_c^{w^{(\varepsilon)}}}
{\|w^{(\varepsilon)}\|_{\ell^2}},
\]
from Theorem \ref{thm:asymptotic-isometry}.
By linearity of the coefficient blocks,
\begin{align*}
&
\left\|
B_\lambda^{(\zeta)}(E_\varepsilon c)
-
s_\lambda(\zeta)B_0^{(\zeta)}(E_\varepsilon c)
\right\|_{\ell^2(K_D)}
\\
&\qquad\leq
d^{-1/4}C 
\frac{
\|T_{-\lambda}w^{(\varepsilon)}
-w^{(\varepsilon)}\|_{\ell^2}
}{
\|w^{(\varepsilon)}\|_{\ell^2}
}
\|c\|.
\end{align*}
Consequently,
\begin{equation}
\label{eq:normalized-block-defect-limit}
\left\|
B_\lambda^{(\zeta)}(E_\varepsilon c)
-
s_\lambda(\zeta)B_0^{(\zeta)}(E_\varepsilon c)
\right\|_{\ell^2(K_D)}
\longrightarrow0
\end{equation}
for every fixed \(\lambda\in\Lambda_D\), \(\zeta\in\mathbb R^{2n}\),
and \(c\in\mathcal H_{\mathbf d}\).

The convergence in
\eqref{eq:normalized-block-defect-limit} is an absolute
block-covariance statement. It does not by itself imply convergence
of the block-correlation quotient, because the norm of the normalized
reference block may also tend to zero. To recover the syndrome
character through block correlation, we use instead the weak-*
approximation of the ideal codeword.

Assume, in addition, that the family
\(\{w^{(\varepsilon)}\}_{\varepsilon>0}\) is uniformly bounded in
\(\ell^\infty(\Lambda_D)\) and that
\[
w^{(\varepsilon)}
\overset{w^*}{\longrightarrow}1
\qquad\text{in }\ell^\infty(\Lambda_D).
\]
Then, by \eqref{eq:weak-star-ideal-limit},
\[
\Psi_c^{w^{(\varepsilon)}}
\overset{w^*}{\longrightarrow}
\psi_c
\qquad\text{in }M^\infty(\mathbb R^n).
\]
Since Weyl operators act continuously on \(M^\infty(\mathbb R^n)\), it follows that
\[
\rho(\zeta)\Psi_c^{w^{(\varepsilon)}}
\overset{w^*}{\longrightarrow}
\rho(\zeta)\psi_c.
\]
Moreover, every coefficient functional
\[
f\longmapsto
a_\mu\bigl(\Phi_{\gamma,\alpha}(f)\bigr)
=
\frac1{d}
\langle f,\rho(\mu)h_\alpha\rangle ,\qquad f\in M^\infty(\mathbb R^n)
\]
is weak-* continuous because \(\rho(\mu)h_\alpha\in M^1(\mathbb R^n)\).
For every fixed \(q\in K_D\), weak-* continuity of the coefficient
functional gives
\begin{align*}
e^{\pi i\sigma(\lambda,\nu_q)}
a_{\lambda+\nu_q}
\left(
\Phi_{\gamma,\alpha}
\bigl(\rho(\zeta)\Psi_c^{w^{(\varepsilon)}}\bigr)
\right)
\longrightarrow
e^{\pi i\sigma(\lambda,\nu_q)}
a_{\lambda+\nu_q}
\left(
\Phi_{\gamma,\alpha}
\bigl(\rho(\zeta)\psi_c\bigr)
\right).
\end{align*}
Thus, the regularized blocks converge coordinatewise to the ideal
block. Since \(K_D\) is finite, coordinatewise convergence is
equivalent to convergence in \(\ell^2(K_D)\), and hence
\[
B_\lambda^{(\zeta),w^{(\varepsilon)}}(c)
\longrightarrow
B_\lambda^{(\zeta)}(\psi_c)
\qquad\text{in }\ell^2(K_D),
\]
for every fixed \(\lambda\in\Lambda_D\).

Suppose that $B_0^{(\zeta)}(\psi_c)\neq0$.
This holds, in particular, when \(c\neq0\), \(v_\alpha\neq0\), and
\(\zeta\) is sufficiently small, by
Corollary~\ref{cor:block-correlation-syndrome}. It follows that $B_0^{(\zeta),w^{(\varepsilon)}}(c)\neq0$
for every sufficiently small \(\varepsilon\). Hence
\begin{equation*}
\frac{
\left\langle
B_\lambda^{(\zeta),w^{(\varepsilon)}}(c),
B_0^{(\zeta),w^{(\varepsilon)}}(c)
\right\rangle
}{
\left\|
B_0^{(\zeta),w^{(\varepsilon)}}(c)
\right\|_{\ell^2(K_D)}^2
}
\longrightarrow
s_\lambda(\zeta)
=
\chi_D(\lambda)e^{2\pi i\sigma(\lambda,\zeta)}.
\end{equation*}
Thus, under the weak-* approximation hypothesis, the
block-correlation estimator for the regularized states converges to
the ideal GKP syndrome character.

\section{Conclusion and outlook}

We have developed a time--frequency framework for lattice GKP codes
in which ideal codewords are naturally realized in
\(M^\infty(\mathbb R^n)\), the dual of the Feichtinger algebra
\(M^1(\mathbb R^n)\). This choice is narrower than the full space of
tempered distributions and excludes derivative-type solutions of the
stabilizer equations. The vector-valued Zak transform consequently
identifies the ideal GKP code bijectively with a finite-dimensional
logical fibre supported at the trivial point of the continuous
syndrome torus.

Within this framework, multi-window Gabor analysis provides stable
localized coordinates indexed by the adjoint lattice. Our main
structural result shows that one fundamental block of these
coefficients is exactly a finite Weyl-frame analysis of the logical
vector. After normalization, the block map is an isometry,
\[
B_0^*B_0=I_{\mathcal H_{\mathbf d}},
\]
and
\[
\Pi_0=B_0B_0^*
\]
is the orthogonal projection of arbitrary coefficient data onto the
space of admissible exact blocks.

We have also constructed normalizable GKP approximants through
lattice-envelope Gabor multipliers. The localized seed carries the
logical coefficients, whereas the multiplier symbol controls decay
along the stabilizer orbit. Envelopes approaching the constant symbol
converge weak-$*$ in \(M^\infty\) to the ideal codeword, and
asymptotic translation invariance yields asymptotically isometric
logical encodings. Finite oscillator energy requires additional regularity.

Finally, displacement syndromes appear as characters relating
translated coefficient blocks. Block correlation eliminates the
unknown logical vector and recovers the syndrome character, with a
deterministic stability estimate under additive coefficient
perturbations. These results concern access to complex
time--frequency coefficients and do not yet constitute a physical
syndrome-extraction protocol.

\section*{Statements and Declarations}

% REQUIRED BEFORE SUBMISSION: authors must complete and approve this statement.
\paragraph{Author contributions.}
Both authors contributed equally to this work.

\paragraph{Use of AI tools.}
ChatGPT (OpenAI) was used to adapt the manuscript's LaTeX formatting to the Quantum document class.

\paragraph{Conflict of interest.}
The authors declare that they have no competing interests.

\paragraph{Data availability.}
Data sharing is not applicable to this article as no datasets
were generated or analysed during the current study.

\end{document}